\documentclass[11pt, letterpaper]{article}
\usepackage{graphicx} 
\usepackage{amsmath, amsthm, amssymb}
\usepackage{hyperref}
\usepackage[left=1in, bottom=1in, right=1in, top=1in]{geometry}

\DeclareMathOperator*{\argmin}{arg\,min}
\usepackage{array}
\usepackage[authoryear, square]{natbib}
\usepackage{xcolor}
\usepackage{soul}
\usepackage{booktabs}
\usepackage{multirow}
\usepackage{subcaption}

\usepackage{placeins}   
\usepackage{makecell}
\newcommand{\sign}{\operatorname{sign}}
\usepackage{comment}

\usepackage{setspace}
\usepackage[ruled,vlined]{algorithm2e}

\usepackage{mathtools}
\usepackage{amsfonts}
\usepackage{mathrsfs}

\usepackage{epsfig}
\usepackage{epstopdf}
\usepackage{adjustbox}
\usepackage{enumitem}
\usepackage{authblk}
\usepackage{cleveref}

\theoremstyle{plain}

\newtheorem{lemma}{Lemma}[section]
\newtheorem{proposition}{Proposition}[section]

\newtheoremstyle{uprightremark}
  {}                             
  {}                             
  {\normalfont}                  
  {}                             
  {\bfseries}                    
  {.}                            
  { }                            
  {}                             
\theoremstyle{uprightremark}
\newtheorem{remark}{Remark}[section]

\theoremstyle{definition}

\newtheorem{definition}{Definition}[section]

\title{Exact Coordinate Descent for High-Dimensional \\
Regularized Huber Regression}
\author[1]{Younghoon Kim$^*$}
\author[2]{Po-Ling Loh}
\author[3]{Sumanta Basu}
\affil[1]{Department of Mathematics, University of Alabama}
\affil[2]{Department of Pure Mathematics and Mathematical Statistics, University of Cambridge}
\affil[3]{Department of Statistics and Data Science, Cornell University}
\def\thefootnote{$*$}\footnotetext{Corresponding author. Email: ykim124@ua.edu}

\begin{document}

\maketitle

\begin{abstract}
This study develops a coordinate descent algorithm for high-dimensional Huber regression with an elastic-net penalty.
Unlike existing gradient descent algorithms and coordinate descent methods based on second-order optimization, the proposed algorithm performs each coordinate update exactly, without approximation or gradient and Hessian computations. 
Consequently, it remains stable and converges rapidly even in the presence of highly correlated covariates or heavy-tailed noise.
Building on these exact updates, adaptive variable screening rules and optimality-condition validation procedures are introduced to determine which coordinates to update at each iteration, substantially accelerating convergence in high-dimensional settings.
Theoretical guarantees are established for both the convergence of the proposed algorithm and the validity of validation procedures.
Extensive simulation studies under heavy-tailed noise and highly correlated designs, together with a real-data application, demonstrate that the proposed method is both accurate and computationally efficient in these challenging settings.
\end{abstract}


\newpage
\section{Introduction}
\label{se:intro}

The theory of robust statistics provides estimators that are resistant to outliers and extreme observations \cite[e.g.,][]{hampel2011robust,huber2011robust,maronna2019robust}. In robust estimation, the squared $\ell_2$ loss is commonly replaced by the Huber loss with threshold $\delta>0$,
\begin{equation}\label{e:huber_loss}
    \rho_{\delta}(u) 
    = \left\{ \begin{array}{ll}
        \frac{1}{2}u^2, & \textrm{if } |u| \leq \delta, \\
        \delta|u| - \frac{1}{2}\delta^2, & \textrm{if } |u| > \delta.
    \end{array} \right.
\end{equation}
Its derivative is given by
\begin{equation}\label{e:huber_deriv}
    \rho_{\delta}'(u) 
    = \left\{ \begin{array}{ll}
        u, & \textrm{if } |u| \leq \delta, \\
        \delta\textrm{sign}(u), & \textrm{if } |u| > \delta,
    \end{array} \right.
\end{equation}
showing that the loss is quadratic for small residuals and linear for large residuals. As a result, the influence of outliers is bounded since the derivative remains constant once $|u|>\delta$.

In the regression setting, this principle extends to coefficient estimation, yielding models that capture the central trend of the response while reducing the influence of atypical observations. In particular, for high-dimensional setups, incorporating sparsity into estimation enhances interpretability and might be expected to lend computational efficiency. Regularized Huber regression naturally bridges between classical linear regression and robust regression by incorporating both sparsity and robustness \cite[e.g.,][]{avella2015robust,loh2024theoretical}. Specifically, for $n$ observation pairs $\{y_i, X_i\}_{i=1,\ldots,n}$, where $X_i$ is a $p$-dimensional vector of covariates and $y_i$ is a scalar response, the coefficient vector $\hat{\beta}$ is estimated by minimizing a linear combination of a robust loss function and an elastic net penalty \citep{zou2005regularization}:
\begin{align}\label{e:adaptive_huber}
    \hat{\beta} &\in \argmin_{\beta \in\mathbb{R}^p}\left\{ F(\beta) := f(\beta) + \lambda\mathcal{R}(\beta) := \frac{1}{n}\sum_{i=1}^n w_i\rho_{\delta}(y_i - X_i^{\top}\beta)+ \lambda\mathcal{R}_{\alpha}(\beta) \right\}, \\
    & \mathcal{R}_{\alpha}(\beta) 
    = \alpha \|\beta\|_1 + \frac{(1-\alpha)}{2}\|\beta\|_2^2, \label{e:elastic_net}
\end{align}
where $\lambda>0$ is the regularization parameter, $\alpha\in[0,1]$ controls the mix between Lasso ($\alpha=1$) and ridge ($\alpha=0$) penalties. The weights $\{w_i\}_{i=1,\ldots,n}$ provide further flexibility in estimation. For simplicity, we set $w_i=1$ for all $i=1,\ldots,n$ throughout this study. Note that the precise form of the robust regression objective in \eqref{e:adaptive_huber} may vary slightly across different contexts \cite[see Section 2.2 in][]{loh2021scale}.

Existing algorithms for regularized robust regression are predominantly based on gradient descent optimization, which requires entire gradient computations (see Section \ref{sse:related}). Because these methods evaluate the gradient at every iteration, they do not fully exploit the sparsity induced by regularization \cite[e.g.,][]{nesterov2013introductory,sun2016majorization,beck2017first}, making them computationally expensive in high-dimensional settings.

Coordinate descent algorithms, by contrast, are well known for their efficiency in sparse high-dimensional regression, particularly with the squared $\ell_2$ loss \cite[e.g.,][]{wright2015coordinate,shi2016primer}, but they have received relatively little attention for regularized robust regression. A notable exception is \cite{yi2017semismooth}, who proposed a coordinate descent algorithm based on a second-order (Newton-type) optimization. However, it can become numerically unstable under ill-conditioned designs, such as highly correlated predictors or heavy-tailed data, where the local curvature is poorly conditioned (see also Section \ref{sse:related}). These limitations motivate a new coordinate descent algorithm that is both computationally efficient and numerically stable for high-dimensional regularized robust regression in the challenging settings.

\subsection{Benchmark Methods}
\label{sse:related}

In this section, we introduce several existing benchmarks for optimizing the regularized Huber regression objective \eqref{e:adaptive_huber}. First, as the most common approach, composite gradient descent \cite[CGD; simply GD;][]{nesterov2013gradient} can be employed; this method has been featured in simulation studies in theoretical work on robust statistics \cite[e.g.,][]{loh2017statistical,loh2021scale}. At the $k^{\text{th}}$ iteration, the algorithm computes
\begin{align}
    \beta^{(k+1)} 
    &\in \argmin_{\beta\in\mathbb{R}^p} \left\{ f(\beta^{(k)}) + \langle \nabla f(\beta^{(k)}), \beta-\beta^{(k)} \rangle + \frac{\eta}{2}\|\beta-\beta^{(k)}\|_2^2 + 
    \lambda\mathcal{R}_{\alpha}(\beta) \right\} \nonumber\\
    &= S_{\lambda\alpha/(\eta+\lambda(1-\alpha))}\left(
     \frac{\beta^{(k)} - \frac{1}{\eta}\nabla f(\beta^{(k)})}{1+\lambda(1-\alpha)/\eta}\right), 
    \label{e:gradient_descent}
\end{align}
where $\eta>0$ is the stepsize and $S_{\tau}(z)$ is the soft-thresholding operator,
\begin{equation}\label{e:soft_threshold}
    \left(S_{\tau}(z)\right)_j 
    = \mathrm{sign}(z_j)\max(|z_j|-\tau,0), \quad j=1,\ldots,p.
\end{equation}
Conventionally, one takes $\eta = L_H$ with $L_H = \frac{1}{n}\mathrm{eig}_{\max}(X^{\top}X)$, where $X$ is the $n\times p$ design matrix and $\mathrm{eig}_{\max}(M)$ denotes the largest eigenvalue of $M$. This is the same Lipschitz constant appearing in Proposition \ref{prop:convergence}.

Similar to the GD algorithm, \cite{sun2020adaptive} applied the iterative local adaptive majorize-minimization method \cite[I-LAMM;][]{fan2018lamm}, a computational framework built on the majorization-minimization (MM) algorithm \cite[see][for a comprehensive survey]{sun2016majorization}. At the $k^{\text{th}}$ iteration, one constructs a majorizing function $h_k(\beta|\beta^{(k)})$ for $f(\beta)$ satisfying
\begin{displaymath}
    h_k(\beta|\beta^{(k)}) \geq f(\beta), \quad h_k(\beta^{(k)}|\beta^{(k)}) = f(\beta^{(k)}),
\end{displaymath}
so that the objective decreases at each step:
\begin{displaymath}
    f(\beta^{(k+1)}) 
    \leq h_k(\beta^{(k+1)}|\beta^{(k)}) 
    \leq h_k(\beta^{(k)}|\beta^{(k)}) = f(\beta^{(k)}).
\end{displaymath}
Taking
\begin{displaymath}
    h_k(\beta|\beta^{(k)}) 
    = f(\beta^{(k)}) + \langle \nabla f(\beta^{(k)}), \beta-\beta^{(k)} \rangle + \frac{\eta_k}{2}\|\beta-\beta^{(k)}\|_2^2,
\end{displaymath}
where $\eta_k$ is chosen adaptively to ensure $h_k(\beta^{(k+1)}|\beta^{(k)})\geq f(\beta^{(k+1)})$, the update takes the same form as \eqref{e:gradient_descent} with $\eta$ replaced by $\eta_k$:
\begin{displaymath}
    \beta^{(k+1)} 
    = S_{\lambda\alpha/(\eta_k+\lambda(1-\alpha))}\left(\frac{\beta^{(k)}-\frac{1}{\eta_k}\nabla f(\beta^{(k)}) }{1+\lambda(1-\alpha)/\eta_k}\right).
\end{displaymath}

As the coordinate descent algorithm for the penalized robust regression, \cite{yi2017semismooth} employed a Newton-type local approximation \cite[e.g.,][]{mifflin1977semismooth,chen2012smoothing} for nonsmooth functions. Specifically, for each coordinate $j=1,\ldots,p$, the Karush–Kuhn–Tucker conditions (KKT) for \eqref{e:adaptive_huber} are given by
\begin{displaymath}
     -\frac{1}{n}\sum_{i=1}^n X_{ij}\rho_{\delta}'(y_i - X_i^{\top}\hat{\beta}) + \lambda \alpha \hat{s}_j + \lambda(1-\alpha)\hat{\beta}_j = 0, \quad
     \hat{\beta}_j - 
    \mathcal{S}_{1}(\hat{\beta}_j+\hat{s}_{j}) 
    = 0,
\end{displaymath}
where $\hat{s}_{j} \in \partial|\hat{\beta}_j|$ is the slack variable for the $j^{\textrm{th}}$ coordinate, and $\partial|u|$ denotes the subdifferential of the absolute value function. For simplicity, we omit the iteration index $k$. Define $u_{ij}(\beta_j) = y_i - \sum_{h=1}^p X_{ih}\tilde{\beta}_h + X_{ij}(\tilde{\beta}_j-\beta_j)$, $i=1,\ldots,n$, as the unscaled partial residual for the $j^{\textrm{th}}$ variable, where $\tilde{\beta}_{h}$ (and $\tilde{s}_{h}$) are $\hat{\beta}_j$ (and $\hat{s}_{j}$) while $h \neq j$ coordinates remains the same from the previous $j^{\textrm{th}}$ coordinate update so that $u_{ij}(\beta_j)$ is the residual that would result from replacing $\tilde\beta_j$ with $\beta_j$ while keeping all other coordinates fixed at $\tilde\beta_h$. Then the current update of the $j^{\textrm{th}}$ coordinate can be written as a function of $(\tilde{\beta}_j,\tilde{s}_j)$:
\begin{displaymath}
    \begin{pmatrix}
        \beta_j \\ s_j
    \end{pmatrix} 
    \leftarrow \begin{pmatrix}
        \tilde{\beta}_j 
        + \frac{\frac{1}{n}\sum_{i=1}^n\rho_{\delta}'(u_{ij}(\tilde{\beta}_j))X_{ij} - \lambda\alpha\textrm{sign}(\tilde{\beta}_j+\tilde{s}_j) - \lambda(1-\alpha)\tilde{\beta}_j}{\frac{1}{n}\sum_{i=1}^n X_{ij}^2\mathbf{1}_{\{|u_{ij}(\tilde{\beta}_{j})| \leq \delta\}} + \lambda(1-\alpha)} \\
        \textrm{sign}(\tilde{\beta}_j + \tilde{s}_j)
    \end{pmatrix},
\end{displaymath}
when $|\tilde{\beta}_j+\tilde{s}_j| \geq 1$; the case when $|\tilde{\beta}_j+\tilde{s}_j| \leq 1$ can be derived similarly. 

The Hessian appearing in the algorithm is
\begin{displaymath}
    H_j := \begin{pmatrix}
        a_j & \lambda \alpha \\ 0 & -1
    \end{pmatrix}, \quad a_j:=\frac{1}{n}\sum_{i=1}^n X_{ij}^2\mathbf{1}_{\{|u_{ij}(\tilde{\beta}_j)| \leq \delta\}},
\end{displaymath}
where $a_j$ represents the average curvature of the Huber loss with respect to the $j^{\textrm{th}}$ variable, restricted to observations whose residuals fall within the quadratic region $[-\delta,\delta]$ of the Huber loss. The condition number $\kappa(H_j)$ is given by
\begin{displaymath}
    \kappa(H_j) 
    = \frac{\sqrt{\textrm{eig}_{\max}(H_j^{\top}H_j)}}{\sqrt{\textrm{eig}_{\min}(H_j^{\top}H_j)}}
    = \sqrt{\frac{c_j + \sqrt{c_j^2-4a_j^2}}{c_j - \sqrt{c_j^2-4a_j^2}}},
\end{displaymath}
where $c_j = a_j^2+(\lambda\alpha)^2+1$. It follows that $\kappa(H_j)\to\infty$ as $a_j\to0$. This reveals two potential sources of instability: First, the convergence of the algorithm \cite[Theorem 3.2 in][]{yi2017semismooth} requires at least one residual to lie in the quadratic region, i.e., $|u_{ij}(\tilde{\beta}_j)|\leq\delta$ for some $i\in\{1,\ldots,n\}$, so that $a_j>0$. If it fails, which can happen when $\delta$ is small or residuals are too large, under heavy-tailed noise, $H_j$ becomes nearly singular, and the Newton step is not well-defined. Even when this condition holds marginally, $a_j\approx 0$ makes $\kappa(H_j)$ large and the Newton step numerically unstable. Also, when the columns of $X$ are highly correlated, the algorithm becomes particularly vulnerable, since its Newton step scales the gradient solely by $a_j^{-1}$. Specifically, when $X_{j}\approx X_{j'}$, the term $\frac{1}{n}\sum_{i=1}^nX_{ij}X_{ij'}\mathbf{1}_{\{|u_{ij}(\tilde{\beta}_j)| \leq \delta\}}$ is omitted from the update. This is precisely the information needed to apportion the signal between $\beta_j$ and $\beta_{j'}$ correctly, and its absence causes the two coordinate updates to be inconsistent, leading to oscillation across iterations. Taken together, the practical performance of the Newton-type method can be highly sensitive under challenging design matrix conditions. Both failure modes trace back to the fact that the update relies on a single linearization around the current curvature estimate $a_j$. As a result, it inherits whatever fragility $a_j$ has, whether $a_j$ is near zero or the cross-coordinate information is lost when columns are collinear.

\subsection{Contributions}
\label{sse:contribution}

We introduce a pathwise exact coordinate descent algorithm (hereafter referred to simply as \emph{coordinate descent}) for estimating parameters in elastic net-regularized Huber regression. Here, the term \emph{exact} indicates that each coordinate update is obtained without approximation and computations of gradients and Hessians. The proposed method is inspired by the pathwise coordinate descent framework for penalized regression with squared $\ell_2$ losses \citep{friedman2007pathwise,mazumder2011sparsenet}. In particular, compared with existing benchmark methods, our algorithm remains effective under ill-conditioned design matrices, such as those with highly correlated predictors or heavy-tailed data.

We establish convergence guarantees and present practical computational strategies, including adaptive variable screening rules and optimality-condition validation procedures, to accelerate the algorithm. The validation procedure is also theoretically analyzed. We further demonstrate its performance through numerical experiments designed to evaluate these challenging scenarios. Finally, we provide an R package, \texttt{rome}, developed similarly to the widely used \texttt{glmnet} package \citep{friedman2010regularization}, making it straightforward for users familiar with \texttt{glmnet} to adopt our software.

\subsection{Organization of the Paper}
\label{sse:organization}

Section \ref{se:method} introduces the essential ideas, motivated by the example of computing a univariate Huber location estimator and extending it to regularized Huber regression. The statement of the convergence guarantee is also presented. Section \ref{se:compuation} presents computational techniques that accelerate the convergence of the algorithm and establishes theoretical guarantees for both the algorithm and the proposed techniques. Section \ref{se:numerical} reports numerical experiments under scenarios in which the proposed method performs well while the benchmark methods fail. In addition, it demonstrates the effectiveness of the computational acceleration techniques and presents a data application under the assumed conditions. Section \ref{se:discussion} concludes the paper with a discussion.


\section{Proposed Method}
\label{se:method}

In this section, we present the proposed algorithm. We first introduce the idea through a simplified problem, the regularized Huberized median, and then extend it to regularized Huber regression. Finally, we establish the convergence rate of the algorithm.

\subsection{Huberized Median for Univariate Samples}
\label{sse:huberized}

We explain the proposed algorithm by illustrating the univariate Huberized median problem, possibly with a penalty term. Suppose we have one-dimensional observations $\{x_1,\ldots,x_n\}$, and $\lambda \geq 0$ and $\alpha\in[0,1]$ are given. The regularized Huberized median is obtained by solving
\begin{equation}\label{e:huber_median}
    \hat{b} 
    \in \argmin_{b\in\mathbb{R}} \left\{ F(b) = \sum_{i=1}^n \rho_{\delta}(x_i - b) + \lambda\left(\alpha|b| + \frac{(1-\alpha)}{2}b^2\right) \right\}.
\end{equation}
Recall the Huber loss \eqref{e:huber_loss} with threshold $\delta>0$, which can be written as
\begin{displaymath}
    \rho_{\delta}(u) = \frac{u^2}{2}\mathbf{1}_{\{|u|\leq\delta\}} + \left(\delta|u| - \frac{\delta^2}{2}\right)\mathbf{1}_{\{|u|>\delta\}}.
\end{displaymath}
Its derivative \eqref{e:huber_deriv} is then
\begin{displaymath}
    \rho_{\delta}'(u) = u\mathbf{1}_{\{|u|<\delta\}} + \delta\sign(u)\mathbf{1}_{\{|u|>\delta\}}.
\end{displaymath}
Using this expression, the derivative of $f(c)$ in \eqref{e:huber_median} with respect to $b$ is
\begin{equation}\label{e:derive_huber_median}
    F'(b) 
    = \sum_{i;|x_i-b|\leq \delta} (b-x_i) + \delta \sum_{i;|x_i-b| > \delta} \sign(b-x_i) 
    + \lambda \left(\alpha\sign(b) + (1-\alpha)b\right).
\end{equation}
We wish to find $b$ such that $F'(b)=0$ in \eqref{e:derive_huber_median}. Among the two summations in \eqref{e:derive_huber_median}, the $i^{\text{th}}$ observation belongs to exactly one, determined by the location of $b$. Specifically, the marginal contribution of the $i^{\text{th}}$ observation to the derivative is
\begin{equation}\label{e:derivative_univ}
    S_i = \left\{\begin{array}{cl}
      -\delta,   &  \textrm{if} \ b < x_i-\delta, \\
      1\times(b-(x_i-\delta)) - \delta,   & \textrm{if} \ x_i-\delta \leq b \leq x_i + \delta,  \\
     \delta, & \textrm{if} \ b > x_i + \delta.
    \end{array}\right.
\end{equation}
Hence, the marginal increment $S_i$ remains $-\delta$ when $b < X_i-\delta$, increases linearly with slope 1 over the interval $(X_i-\delta, X_i+\delta)$, and reaches $\delta$ once $b > X_i+\delta$.

The structure of the marginal contribution in \eqref{e:derivative_univ} implies that as $b$ varies, the incremental change can be determined simply by counting how many observations have $b$ within their intervals $(x_i-\delta,x_i+\delta)$. The increment is then proportional to the number of such observations relative to the total. A straightforward way to locate the potential value of $b$ is to use grid search: construct a fine grid from $\min_i(x_i-\delta)$ to $\max_i(x_i+\delta)$, and trace the trajectory of $F'(b)$ to identify where it crosses zero. Obviously, this trajectory is monotonically increasing.

Suppose we construct a grid of $2n$ points from the observations and thresholds, $\{x_1-\delta,x_1+\delta,\ldots,x_n-\delta,x_n+\delta\}$. Sorting these values yields the ordered sequence $v_1\leq v_2\leq\cdots\leq v_{2n}$, called \emph{kinks}. For a given $b$, define the active set $\mathcal{I}(b)=\{i:|b-x_i|\leq\delta\}$ to be the collection of indices whose corresponding intervals $(x_i-\delta,x_i+\delta)$ contain $b$. From \eqref{e:derivative_univ}, each observation $i$ contributes slope 1 to $F'(b)$ when $i \in \mathcal{I}(b)$, and slope 0 otherwise. Since $x_i-\delta< x_i + \delta$ always holds, sorting the $n$ left endpoints $\{x_i-\delta\}_{i=1}^n$ separately retains their original indices $i_1,\ldots,i_n$. As $b$ increases through the sorted kinks, observation $i_s$ enters $\mathcal{I}(b)$ when $b$ crosses the $s^{\text{th}}$ left endpoint, incrementing the cumulative slope $A$ by 1. Between any two consecutive kinks $v_m$ and $v_{m+1}$, $\mathcal{I}(b)$ remains constant and $F'(b)$ is affine in $b$ with slope
\begin{displaymath}
    A_m = |\mathcal{I}(v_m)|+\lambda(1-\alpha),
\end{displaymath}
where the $\lambda(1-\alpha)$ term comes from the squared penalty. The value of $F'(b)$ is propagated incrementally from one kink to the next as $F'(v_{m+1})=F'(v_m)+A_m(v_{m+1}-v_m)$. Since $F'(b)$ is monotonically increasing, there exists a unique index $m$ such that $F'(v_m)<0\leq F'(v_{m+1})$. In the interval $(v_m,v_{m+1})$, $F'(b)$ is affine with slope $A_m>0$, and the root is obtained analytically as
\begin{displaymath}
    \hat{b} = v_m-\frac{F'(v_m)}{A_m}.
\end{displaymath}
Note that the $|b|$ component in the penalty term introduces an additional $\lambda\alpha$ adjustment around $b=0$, encouraging the optimal solution $\hat{b}$ to shrink toward zero via soft-thresholding, while the $b^2$ component provides an additional $(1-\alpha)$ slope with respect to $b$, which makes the $F'(b)$ trajectory smoother by increasing the marginal increments. Together, these two effects are analogous to the elastic net operation.

\begin{figure}[t]
     \begin{center}
     \includegraphics[width=0.48\textwidth,height=0.25\textheight]
     {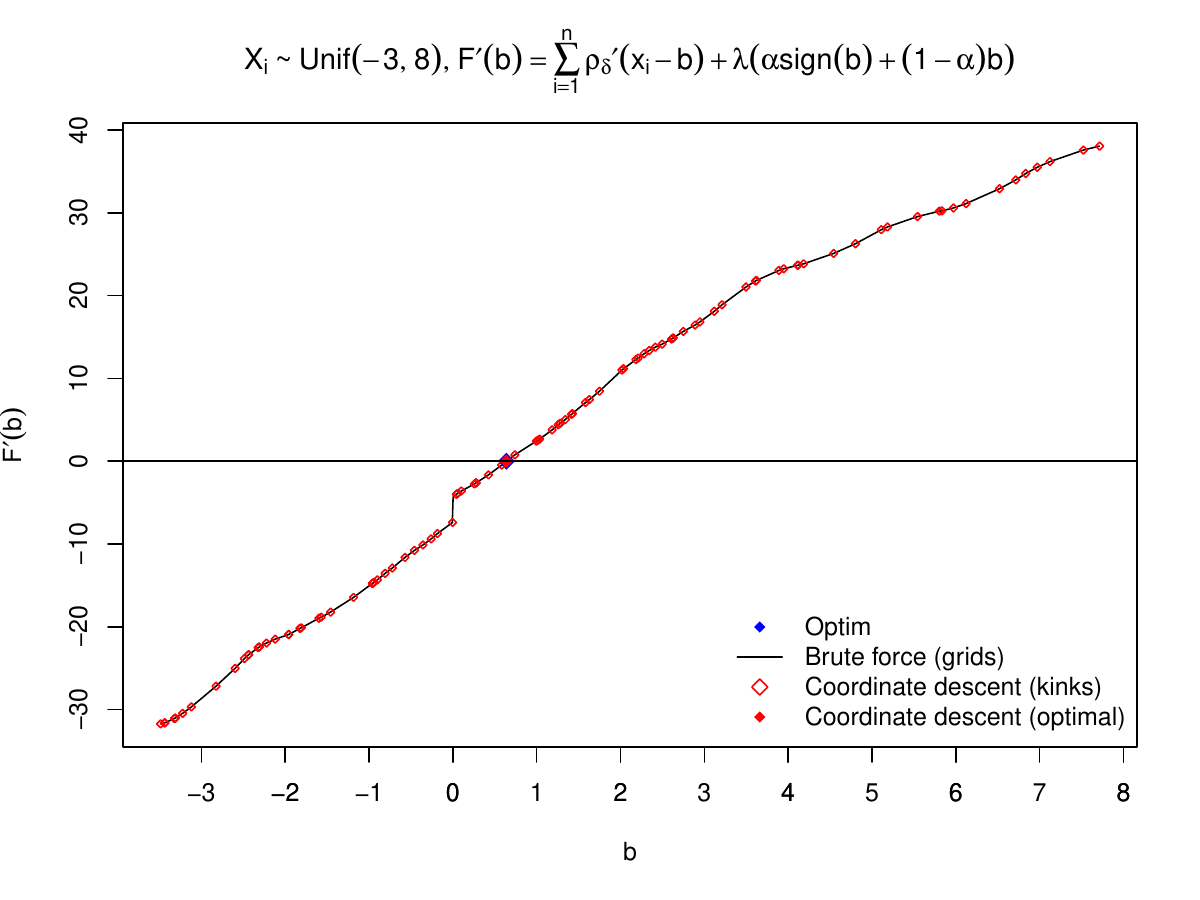}
    \includegraphics[width=0.48\textwidth,height=0.25\textheight]
     {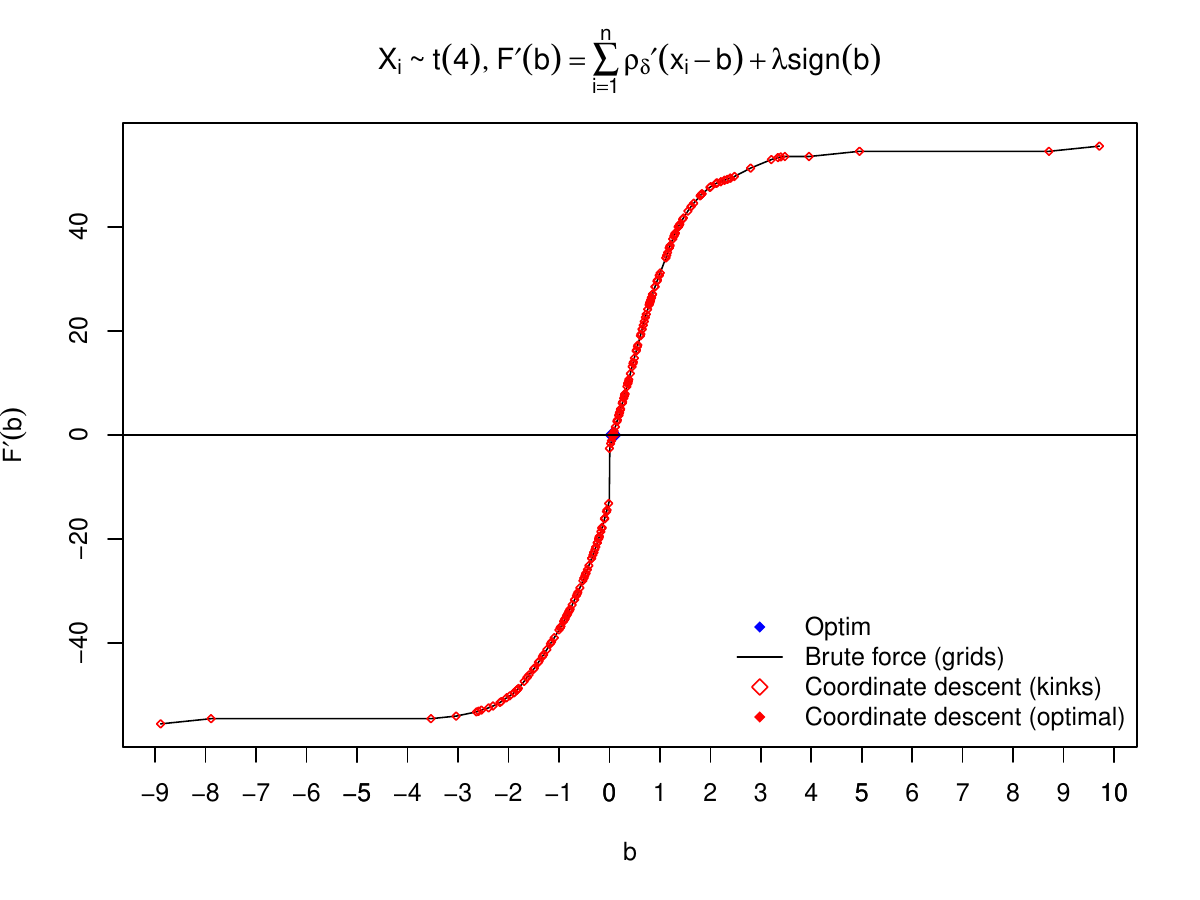}
     \end{center}
    \caption{Illustration of the Huberized medians for Data I (left) and Data II (right) under different values of $\lambda$ and $\alpha$. The plots display the solutions obtained by the coordinate descent algorithm (red dot), the ordered kinks used for the algorithm (red diamonds), the results of a grid search evaluating  $F'(b)$ over a range of $b$ values (black line), and the solution from a univariate optimization solver applied to $F(b)$ (blue dot).}
    \label{fig:hubermized_median}
\end{figure}

\subsubsection{Two Examples}
\label{ssse:examples}

To visualize the algorithm, we present two examples that illustrate different cases. In the first example, suppose we draw 50 random samples from a uniform distribution on $(-3,8)$ (Data I). In this case, the Huberized median is unlikely to lie at zero. The left panel of Figure \ref{fig:hubermized_median} shows the grid search results (black lines) with the kinks marked in red, all of which lie exactly on the constructed grid. Because the Huberized median is located away from zero, the zero-crossing of $F'(c)$ also occurs away from zero. Consequently, introducing an elastic net penalty term with $\lambda=3$ and $\alpha=0.5$ does not shrink the optimal solution toward zero.

In the second example, we consider 100 random samples from a $t$-distribution with 4 degrees of freedom (Data II). Aside from a few outliers, most observations are symmetrically distributed; thus, the Huberized median is likely to fall near zero. The right panel of Figure \ref{fig:hubermized_median} displays the results using the same representation as the first example, with all kinks again lying exactly on the constructed grid. Compared to Data I, these samples exhibit larger values. Even in this case, however, adding the Lasso penalty term with $\lambda=5$ ($\alpha=0$) shrinks the solution exactly to zero.

\subsection{Extension to Coordinate Descent for Regularized Huber Regression}
\label{sse:regression}

We extend the idea of the Huberized median introduced in Section \ref{sse:huberized} to derive the estimator for regularized Huber regression via a coordinate descent algorithm.

While the \texttt{rome} package can also accommodate adaptive Huber regression \cite[e.g.,][]{loh2021scale}, the approach throughout this study focuses on the unweighted case. Accordingly, we assume $w_i=1$ for all $i=1,\ldots,n$. 

In this section, we assume that the tuning parameters $\lambda$ and $\alpha$ are fixed. We use the superscript $(k,j)$ to represent the $k^{\textrm{th}}$ iteration during which the $j^{\textrm{th}}$ coordinate update is complete.

Recall the regularized Huber regression objective in \eqref{e:adaptive_huber}:
\begin{equation}\label{e:huber_regression}
    \hat{\beta} \in \argmin_{\beta \in\mathbb{R}^p} \left\{ F(\beta) := \frac{1}{n}\sum_{i=1}^n \rho_{\delta}(y_i-X_{i}^{\top}\beta) + \lambda\left(\alpha\|\beta\|_1 + \frac{(1-\alpha)}{2}\|\beta\|_2^2\right) \right\}.
\end{equation}
At the $k^{\textrm{th}}$ iteration, for the $j=1,\ldots,p$ variable, the coordinate descent algorithm solves
\begin{equation}\label{e:huber_regression_partial}
    \hat{\beta}_j^{(k,j)} \in \argmin_{\beta_j\in\mathbb{R}}  \left\{ F(\beta_j)
    =: f(\beta_j) + \lambda\mathcal{R}_{\alpha}(\beta_j) = \frac{1}{n}\sum_{i=1}^n \rho_{\delta}\left(X_{ij}(r_{ij}^{(k,j-1)} - \beta_j)\right) + \lambda \left(\alpha|\beta_j| +\frac{(1-\alpha)}{2}\beta_j^2 \right) \right\},
\end{equation}
where $r_{ij}^{(k,j-1)} = r_{ij}(\beta_j^{(k,j-1)}) =\frac{y_i - \sum_{h\neq j}X_{ih}\hat{\beta}_h^{(k,j-1)}}{X_{ij}}$, $i=1,\ldots,n$, is the partial residual of the $j^{\text{th}}$ variable from the previous updates up to the $(j-1)^{\textrm{th}}$ coordinate. Similar to \eqref{e:derive_huber_median}, the derivative of the penalized regression with respect to the $j^{\text{th}}$ variable is 
\begin{align*}
    F'(\beta_j^{(k,j)}) 
    &= \frac{1}{n}\sum_{i;|r_{ij}^{(k,j-1)}-\beta_j|\leq \delta/|X_{ij}|} X_{ij}^2(\beta_j-r_{ij}^{(k,j-1)}) \\
    & + \frac{\delta}{n}\sum_{i;|r_{ij}^{(k,j-1)}-\beta_j| > \delta/|X_{ij}|}|X_{ij}|\sign(\beta_j-r_{ij}^{(k,j-1)}) + \lambda \left(\alpha\sign(\beta_j) + (1-\alpha)\beta_j\right).
\end{align*}
Now, for each iteration, the term $r_{ij}^{(k,j)}$ at the update of the $j^{\text{th}}$ variable plays the role of observations at the univariate Huberized median in Section \ref{sse:huberized}. The marginal contribution of the $i^{\text{th}}$ observation to $f'(\beta_j)$ is
\begin{displaymath}
    S_i^{(k,j)} = \left\{\begin{array}{cl}
      -\frac{|X_{ij}|}{n},   &  \textrm{if} \ \beta_j < r_{ij}^{(k,j-1)} - \delta/|X_{ij}|, \\
      \frac{X_{ij}^2}{n}\times(\beta_j-(r_{ij}^{(k,j-1)}) - \delta/|X_{ij}|)) - \frac{|X_{ij}|}{n},   & \textrm{if} \ r_{ij}^{(k,j-1)} - \delta/|X_{ij}| \leq \beta_j \leq r_{ij}^{(k,j-1)} + \delta/|X_{ij}|,  \\
     \frac{|X_{ij}|}{n}, & \textrm{if} \ \beta_j > r_{ij}^{(k,j-1)} + \delta/|X_{ij}|.
    \end{array}\right.
\end{displaymath}
Analogous to the case of the Huberized median, we obtain $2n$ kinks by sorting the set $$\left\{r_{ij}^{(k,j-1)}\pm \frac{\delta}{|X_{ij}|}\right\}_{i=1,\ldots,n}$$ in increasing order. For a given $\beta_j$, define the active set $$\mathcal{I}^{(k,j-1)}(\beta_j)=\left\{i : |\beta_j-r_{ij}^{(k,j-1)}|\leq \frac{\delta}{|X_{ij}|}\right\}.$$ From the expression of $S_i^{(k,j)}$, each observation $i$ contributes slope $X_{ij}^2/n$ to $F'(\beta_j)$ when $i\in\mathcal{I}^{(k,j-1)}(\beta_j)$, and slope 0 otherwise. Since $r_{ij}^{(k,j-1)}-\delta/|X_{ij}|<r_{ij}^{(k,j-1)}+\delta/|X_{ij}|$ always holds, storing the $n$ left endpoints $\{r_{ij}^{(k,j-1)}-\delta/|X_{ij}|\}_{i=1}^n$ separately retains their original indices $i_1,\ldots,i_n$. As $\beta_j$ increases through the sorted kinks, observation $i_s$ enters $\mathcal{I}^{(k,j-1)}(\beta_j)$ when $\beta_j$ crosses the $s^{\text{th}}$ left endpoint, incrementing the cumulative slope $A^{(k,j)}$ by $X_{i_s j}^2/n$. Since $X_{i_s j}^2/n\geq0$, the slope $A$ is nondecreasing as the kinks are traversed. Between any two consecutive kinks $v_m^{(k,j)}$ and $v_{m+1}^{(k,j)}$, the active set $\mathcal{I}^{(k,j-1)}(\beta_j)$ remains constant and $F'(\beta_j;\mathcal{I}^{(k,j-1)})$ is affine in $\beta_j$, with cumulative slope
\begin{displaymath}
    A^{(k,j)} = \frac{1}{n}\sum_{i \in \mathcal{I}^{(k,j-1)}(v_m^{(k,j)})} 
    X_{ij}^2 + \lambda(1-\alpha).
\end{displaymath}
The value of $F'$ is propagated incrementally from one kink to the next as 
$$F'(v_{m+1}^{(k,j)};\mathcal{I}^{(k,j-1)})=F'(v_m^{(k,j)};\mathcal{I}^{(k,j-1)})+A^{(k,j)}(v_{m+1}^{(k,j)} - v_m^{(k,j)}).$$ Since $F'(\beta_j;\mathcal{I}^{(k,j-1)})$ is monotonically increasing, there exists a unique index $m$ such that $$F'(v_m^{(k,j)};\mathcal{I}^{(k,j-1)})<0\leq F'(v_{m+1}^{(k,j)};\mathcal{I}^{(k,j-1)}).$$ In the interval $(v_m^{(k,j)}, v_{m+1}^{(k,j)})$, $F'(\beta_j;\mathcal{I}^{(k,j-1)})$ is affine with slope $A^{(k,j)}>0$, and the minimizer is obtained analytically as
\begin{displaymath}
    \hat{\beta}_j^{(k,j)} = v_m^{(k,j)} - 
    \frac{F'(v_m^{(k,j)};\,\mathcal{I}^{(k,j-1)})}{A^{(k,j)}}.
\end{displaymath}

The remaining steps of the algorithm follow those of \texttt{glmnet} \citep{friedman2010regularization}, in the sense that we employ pathwise updates along a decreasing sequence of regularization parameters, use warm starts to accelerate convergence, and apply a screening rule to restrict updates to variables that are likely to be nonzero \cite[see also Chapter 5.4 in][]{hastie2015statistical}.

\begin{remark} 
To the best of our knowledge, there is no practical rule to obtain the optimal weights $\{w_i\}_{i=1,\ldots,n}$ and threshold $\delta$. In terms of the weights, we adopt Mallows-type reweighting to improve robustness against high-leverage points by modifying the Huber loss with covariate-dependent weights. Specifically, in the adaptive regularized Huber regression \eqref{e:adaptive_huber}, we set $w_i=w(X_i)=\min\{1,b/\|X_i\|_2\}$ for some constant $b$. The weight $w(X_i)$ adjusts according to the leverage of the $i^{\text{th}}$ observation, thereby reducing the influence of covariates with unusually large norms \cite[see also Section 3.2 in][]{loh2017statistical}. From an algorithmic perspective, this is equivalent to replacing $(y_i,X_{ij})$ with $(w_iy_i,w_iX_{ij})$, so the computational procedure remains unchanged.

The choice of the threshold parameter $\delta$ is more complicated due to its relationship to the scale of the additive errors, usually assumed to be unknown a priori \cite[e.g., Chapter 7.7 in][]{huber2011robust}. Combining our approach with methods such as optimizing Huber's concomitant estimator to obtain the scale parameter is beyond the scope of our current analysis, and is left for future research \cite[see Section 3 in][for a related discussion]{loh2021scale}.
\end{remark}

\subsection{Theory of Convergence}
\label{sse:conv}

This section presents a convergence guarantee for the algorithm proposed in Section \ref{sse:regression}.

\begin{proposition}[Convergence of the algorithm]\label{prop:convergence}
Let $\beta^{(k)} := \beta^{(k,p)}$ denote the solution after the $k^{\text{th}}$ full sweep over all $p$ coordinates. For any minimizer $\beta^*$ of $F(\beta)$ in \eqref{e:adaptive_huber} and for all $k\geq1$, we have
\begin{equation}\label{e:convergence_guarantee}
    F(\beta^{(k)}) - F(\beta^*) \leq \frac{L_H\|\beta^{(0)} - \beta^*\|_2^2}{2k},
\end{equation}
where the Lipschitz constant $L_H$ of the gradient of the Huber loss is bounded by $L_H \leq \frac{1}{n}\mathrm{eig}_{\max}(X^{\top}X)$. Hence, the optimality gap $F(\beta^{(k)}) - F(\beta^*)$ converges to $0$ as $k \rightarrow \infty$.
\end{proposition}
The proof is deferred to the Supplementary Material. Note that the convergence rate $O(1/k)$ commonly appears in gradient descent-type methods \cite[e.g.,][]{nesterov2013introductory,beck2017first}. However, it can be shown that the convergence rate of the coordinate descent algorithm is \emph{faster} than that of the gradient descent algorithms explained in Section \ref{se:intro}.

Our proof follows the approach of \cite{saha2013nonasymptotic}, by comparing the convergence rates of GD, the proposed coordinate descent (CD), and a coordinate gradient descent (CGD; e.g., \cite{tseng2009coordinate}) framework, as detailed in the Supplementary Material. In Section \ref{sse:experiment_convergence}, we provide numerical comparisons of the convergence rates among these methods.

\section{Computational Advancements}
\label{se:compuation}

In this section, we present several computational strategies designed to improve the rate of convergence of our algorithm. We also establish theoretical guarantees to ensure the convergence of the proposed method. Throughout this section, we omit the pair of indices $(k,j)$.

\subsection{Checking Optimality Conditions}
\label{sse:kkt}

In the algorithm described in Section \ref{se:method}, Quicksort is used for reordering kinks. Quicksort is a widely used sorting algorithm that follows a divide-and-conquer strategy: it selects a pivot, partitions the array into elements less than and greater than the pivot, and recursively sorts the subarrays. On average, Quicksort  reorders $n$ elements in $O(n\log n)$ arithmetic operations, which results from performing $O(n)$ work per partition step across $\log n$ levels of recursion. In the worst case, when the pivot choices lead to highly unbalanced partitions, the complexity can degrade to $O(n^2)$ \cite[e.g., see Chapter 7 in][]{cormen2022introduction}. However, this scenario is unlikely in our context, since the order of half of the $2n$ kink points is predetermined.

In the context of the regularized Huber regression algorithm, the arithmetic cost per coordinate update is dominated by sorting the $2n$ kinks generated from the partial residuals $\{r_{ij}\}_{i=1,\ldots,n}$ for each $j^{\text{th}}$ variable. The remaining computations, such as computing cumulative slopes, locating the root where the derivative crosses zero, and updating the global residual vector, require only $O(n)$ operations. Consequently, each coordinate update has complexity $O(n\log n)$, and a full sweep across all variables requires $O(pn\log n)$ arithmetic operations. This is unusual for coordinate descent algorithms, as the computational cost typically scales with the number of variables $p$, not with the number of samples $n$ \cite[e.g.,][]{richtarik2014iteration}. Therefore, it is necessary to develop strategies that reduce the number of variables updated in each iteration.

We use an idea inspired by \cite{kim2025pathwise} (see Proposition 3.1 therein). In quantile regression, the condition relies only on the signs of $X_{ij}$ in the design matrix: once a sign changes, it remains constant thereafter. By contrast, in Huber regression, one must verify whether each observation’s marginal contribution is increasing; i.e., whether the candidate solution lies within the interval defined by the partial residuals and the threshold. Consequently, the condition to be checked is more complex, but the following results show that a similar argument holds in a more general sense.

\begin{proposition}[Karush-Kuhn-Tucker (KKT) condition]\label{prop:KKT}
For the $j^{\text{th}}$ variable, the following conditions characterize the update:
\begin{enumerate}
    \item[(a)] For \eqref{e:huber_regression_partial}, define $\mathcal{I}:=\mathcal{I}(\beta_j)$ as the set of indices $i=1,\ldots,n$ that satisfy $|\beta_j-r_{ij}|\leq \delta/|X_{ij}|$. Then the derivative of $F(\beta)$ with respect to $\beta_j$ can be written as
    \begin{equation}\label{e:slope}
        F'(\beta_j;\mathcal{I}) 
        =
        S_0+\frac{1}{n}\sum_{i\in\mathcal{I}(\beta_j)}\left( 
        X_{ij}^2(\beta_j-r_{ij})+\delta|X_{ij}|\right) 
        +\lambda\left(\alpha\sign(\beta_j) + (1-\alpha)\beta_j\right),
    \end{equation}
    where $S_0=-\frac{\delta}{n}\sum_{i=1}^n |X_{ij}|$.
    \item[(b)] Assume that the kinks $\{v_m\}_{m=1,\ldots,2n}$ are distinct. Let $i_1,\ldots,i_n$ be the observation indices ordered by their left endpoints $r_{i_s j}-\delta/|X_{i_s j}|$, $s=1,\ldots,n$. The nontrivial minimizer $\hat{\beta}_j \in (v_m, v_{m+1})$ exists and is unique if and only if
    \begin{displaymath}
        F'(v_m;\mathcal{I})F'(v_{m+1};\mathcal{I}')<0,
    \end{displaymath}
    for consecutive kinks $v_m<v_{m+1}$, where $\mathcal{I}\subset\mathcal{I}'$ and $|\mathcal{I}'\setminus\mathcal{I}|=1$, with $\mathcal{I}'=\mathcal{I}\cup\{i_s\}$ for the observation $i_s$ whose left endpoint $r_{i_s j}-\delta/|X_{i_s j}|$ equals $v_m$. In this case,
    \begin{displaymath}
        \hat{\beta}_j = v_m - \frac{F'(v_m;\mathcal{I})}{A_m},
    \end{displaymath}
    where $A_m = \frac{1}{n}\sum_{i\in\mathcal{I}'}X_{ij}^2 + \lambda(1-\alpha)$ is the cumulative slope on $(v_m,v_{m+1})$.
    \item[(c)] The current value $\hat{\beta}_j$ is already the minimizer of \eqref{e:huber_regression_partial} if and only if
    \begin{equation}\label{e:KKT_check}
        |F'(\hat{\beta}_j;\mathcal{I}(\hat{\beta}_j))| \leq \lambda\alpha,
    \end{equation}
    where $F'(\hat{\beta}_j;\mathcal{I}(\hat{\beta}_j))$ is given by \eqref{e:slope}. If $\hat{\beta}_j \neq 0$ and the condition \eqref{e:KKT_check} holds, no update is performed. Otherwise, $\hat{\beta}_j =0$ if the condition holds.
\end{enumerate}
\end{proposition} 

The proof is deferred to the Supplementary Material. This step is performed prior to each coordinate update. The partial residuals computed for the KKT optimality condition check \eqref{e:KKT_check} are then reused in the update, if necessary. To demonstrate the efficiency of checking the optimality condition \eqref{e:KKT_check} before updating variables, we compare computation times with and without this check in Section \ref{sse:experiment_kkt}.

\subsection{Warm Starts}
\label{sse:warm}

In practice, it is common to compute a sequence of regularized solutions over a decreasing grid of regularization parameters $\{\lambda_{\ell}\}_{\ell=0,1,\ldots,\ell_{\max}}$. In \texttt{rome}, we set $\ell_{\max}=100$ by default. The grid begins at the maximum regularization parameter value $\lambda_0$, which is derived from the subdifferential of the penalty at zero and given by
\begin{equation}\label{e:lambda_max}
    \lambda_0 = \frac{1}{n\alpha} \max_{j=1,\ldots,p} \left| \sum_{i=1}^n X_{ij}\rho_{\delta}'(y_i) \right|,
\end{equation}
where $\rho_{\delta}'(u)$ is the derivative of the Huber loss \eqref{e:huber_deriv}. This choice is consistent with the standard elastic net convention of selecting the smallest value for which all penalized coefficients are shrunk exactly to zero. Following the approach of the \texttt{glmnet} package, the grid is constructed by decreasing $\lambda$ from $\lambda_0$ down to $\lambda_{\ell_{\max}}=\epsilon \lambda_0$ on a logarithmic scale, with the default in \texttt{rome} set to $\epsilon=0.001$. Each solution $\hat{\beta}(\lambda_{\ell})$ then serves as a warm start for computing $\hat{\beta}(\lambda_{\ell+1})$. Consequently, the number of active, nonzero coefficients tends to increase gradually as $\ell$ increases.

\subsection{Adaptive Sequential Strong Rules}
\label{sse:screen}

Along with the warm start described in Section \ref{sse:warm}, another widely used computational strategy in coordinate descent algorithms is the sequential strong (screening) rule, which identifies a subset of variables likely to be active. The most prominent example is the sequential strong rule \citep{tibshirani2012strong}. However, as noted by \cite{yi2017semismooth}, the conventional setting used in Lasso programs is often unsuitable, as violations are frequently observed. To address this, they proposed a variant known as the adaptive sequential strong rule. 

Define $c_j(\lambda) = -\frac{1}{n}\sum_{i=1}^n X_{ij}\rho_{\delta}'(y_i - X_i^{\top}\hat{\beta}(\lambda))$. Then one has
\begin{displaymath}
    |c_j(\lambda)| 
    \leq \lambda\alpha - \frac{\|X_i\|_2\|y\|_2}{\lambda_0}(\lambda_0-\lambda) 
    = \lambda\alpha - \frac{\|X_i\|_2\|y\|_2}{\frac{1}{n\alpha}\max_i|X_{i}^{\top}\rho_{\delta}(y_i)|}(\lambda_0-\lambda) := \lambda\alpha + M\alpha(\lambda - \lambda_0).
\end{displaymath}
In order for each $c_j$ to be Lipschitz continuous, we need
\begin{displaymath}
    |c_{j}(\lambda) - c_{j}(\lambda')| \leq M\alpha(\lambda - \lambda').
\end{displaymath}
The sequential strong rule of \cite{tibshirani2012strong} is obtained by replacing $\lambda$ with $\lambda_{\ell}$, and $\lambda_{\max}$ and $\lambda'$ with $\lambda_{\ell-1}$, respectively, and setting $M(\lambda_{\ell})=M_{\ell}=1$ for all $\ell$. The adaptive sequential strong rule of \cite{yi2017semismooth} allows $M_{\ell}$ to vary with $\lambda_{\ell}$. In our algorithm, we adopt the latter version without further modification. However, note that the KKT condition in Step 2 below is described in Section \ref{sse:kkt}, and it differs from the one given therein. 

In summary, take $M_0=1$. Then for $\ell=1,\ldots,\ell_{\max}$:
\begin{enumerate}
    \item Derive the eligible set of predictors as $\mathcal{E} = \{j:|c_j(\lambda_{\ell-1})| \geq \alpha\left(\lambda_{\ell} + M_{\ell-1}(\lambda_\ell - \lambda_{\ell-1})\right)\}$, where $M_{0}=1$ and
    \begin{equation}\label{e:adaptive_slope}
        M_\ell = \frac{\max_{j=1,\ldots,p}|c_j(\lambda_{\ell-1}) - c_{j}(\lambda_{\ell})|}{\alpha(\lambda_{\ell-1} - \lambda_{\ell})}.
    \end{equation}
    \item Solve the problem \eqref{e:huber_regression_partial} using only the predictors $j\in\mathcal{E}$. Check the KKT conditions for $j\notin\mathcal{E}$ described in \eqref{e:KKT_check}, and include any $j$ that violates the KKT conditions back into $\mathcal{E}$.
    \item Compute $M_{\ell}$.
\end{enumerate}
We conduct a numerical experiment comparing the two rules: the strong screening rule and its adaptive version described in this section. The results are presented in Section \ref{sse:experiment_rules}.

In \texttt{rome}, the KKT condition checks in Section \ref{sse:kkt} are performed for variables outside the set $\mathcal{E}$ at Step 2 of the adaptive sequential strong rule in Section \ref{sse:screen}. But in principle, they can be applied independently of the screening rule; if the screening rule is not applied, one can simply set $\mathcal{E}=\emptyset$ so that $\mathcal{E}^{c}=\{1,2,\ldots,p\}$. In this case, the algorithm updates all variables at every iteration.

Algorithm \ref{alg:huber_cd_full} summarizes the coordinate descent procedure from Section \ref{sse:regression}, incorporating the computational strategies for the pathwise scheme in Sections \ref{sse:kkt}--\ref{sse:screen}.

\section{Numerical Examples}
\label{se:numerical}

In this section, we evaluate the numerical performance of the proposed algorithm, illustrate its convergence rate empirically, and demonstrate the computational advancements implemented in the package. Throughout the studies using synthetic data, we focus on the Lasso penalty by setting $\alpha=1$. Finally, we illustrate the practical utility of the algorithm by applying it to a real-world dataset.

\subsection{Comparison with Benchmark Methods}
\label{sse:experiment_benchmark}

We conduct numerical experiments to evaluate the performance of the proposed methods by comparing them with existing benchmarks. Specifically, we consider the two recent alternative methods described in Section \ref{se:intro}. First, the R implementation of the MM-based algorithm (ILAMM) is obtained from the author's GitHub repository\footnote{\href{https://github.com/XiaoouPan/ILAMM}{https://github.com/XiaoouPan/ILAMM}}. Second, for the coordinate descent algorithm based on a second-order local approximation, we utilize the R package \texttt{hqreg}, available on the Comprehensive R Archive Network (CRAN).

We generate data and compute the entire trajectory of the sequence $\{\lambda_{\ell}\}_{\ell=0,1,\ldots,100}$, using the same values for all methods. In this study, we compare the computation time (runtime) required to complete all computations for $\lambda_{1},\ldots,\lambda_{100}$ values, as well as the log-scaled root mean square error across the $\lambda$ values. The threshold $\delta$ is fixed at 0.5 and applied uniformly to all methods. Since the definition of the Huber loss function in \texttt{hqreg} differs slightly (by a factor of $\delta$), we adjust the values accordingly when using the $\lambda$ sequence.

The sparse model parameter is inspired by the setting in \cite{gu2018admm}. In addition, we consider columns of design matrices sampled from heavy-tailed distributions and/or with highly correlated covariance structures, possibly with a block structure. Specifically, the $p$-dimensional true coefficient vector $\beta$ is set as
\begin{displaymath}
    \beta =(2,0,1.5,0,0.8,0,1,0,1.75,0,0,0.75,0,0,0.3,\boldsymbol{0}_{p-16}^{\top})^{\top}\in\mathbb{R}^p.
\end{displaymath}
%

{\footnotesize
\begin{algorithm}[H]
\caption{Pathwise coordinate descent for regularized Huber regression}
\label{alg:huber_cd_full}
\DontPrintSemicolon
\SetKwInOut{Input}{Input}\SetKw{Init}{Initialize}
\Input{$\{y_i, X_i\}_{i=1,\ldots,n}$, $\delta, \alpha$, $\{\lambda_{\ell}\}_{\ell=0,1,\ldots,\ell_{\max}}$.}
\Init $\hat{\beta}^{(0,0)}=0$, $M_0=1$, and set $r_i^{(0,0)}=y_i$ for $i=1,\ldots,n$\;
\For{$\ell = 1$ \KwTo $\ell_{\max}$}{
    \tcp{Warm start from previous solution (Section \ref{sse:warm})}
    Set $\hat{\beta}^{(0,0)} \leftarrow \hat{\beta}(\lambda_{\ell-1})$, $r_i^{(0,0)} \leftarrow y_i-X_i^{\top}\hat{\beta}^{(0,0)}$ for $i=1,\ldots,n$ and compute $c_j(\lambda_{\ell-1})=-\frac{1}{n}\sum_{i=1}^n X_{ij}\rho_{\delta}'(r_i^{(0,0)})$ for $j=1,\ldots,p$\;
    \tcp{Step 1: Adaptive sequential strong rule (Section \ref{sse:screen})}
    Compute the eligible set $\mathcal{E}=\left\{j: |c_j(\lambda_{\ell-1})|\geq \alpha\left(\lambda_\ell + M_{\ell-1}(\lambda_\ell - \lambda_{\ell-1})\right)
    \right\}$\;
    Set $k \leftarrow 1$\;
    \Repeat{convergence of $\hat{\beta}^{(k,p)}$}{
        \tcp{Step 2: Coordinate updates over the eligible set}
        \For{$j \in \mathcal{E}$}{
            Compute partial residuals $r_{ij}^{(k,j-1)}=\frac{y_i-\sum_{h\neq j}X_{ih}\hat{\beta}_h^{(k,j-1)}}{X_{ij}}$ for $i=1,\ldots,n$\;
            \tcp{KKT optimality check (Section~\ref{sse:kkt})}
            Evaluate $F'(0;\mathcal{I}^{(k,j-1)}(0))$ using \eqref{e:slope} with residuals $r_{ij}^{(k,j-1)}$\;
            \eIf{$|F'(0\,;\,\mathcal{I}^{(k,j-1)}(0))| \leq \lambda_{\ell}\alpha$}{
                $\hat{\beta}_j^{(k,j)} \leftarrow \hat{\beta}_j^{(k-1,j)}$\;
            }{
                \tcp{Exact update via kink search (Section \ref{sse:regression})}
                Form and sort $2n$ kinks to obtain $v_1^{(k,j)}\leq v_2^{(k,j)}\leq\cdots\leq v_{2n}^{(k,j)}$, where the $n$ left endpoints $\{r_{ij}^{(k,j-1)}-\delta/|X_{ij}|\}_{i=1}^{n}$ are stored separately with their original indices retained as $i_1,\ldots,i_n$\;
                Set $S_0^{(k,j)} = -\frac{\delta}{n}\sum_{i=1}^n |X_{ij}|$, $F'(v_1^{(k,j)}) \leftarrow S_0^{(k,j)}+\lambda\alpha$, $A^{(k,j)} \leftarrow \lambda(1-\alpha)$, and $s \leftarrow 1$\;
                \For{$m = 1$ \KwTo $2n-1$}{
                    \lIf{$v_m^{(k,j)}$ is a left endpoint}{
                        $A^{(k,j)} \leftarrow A^{(k,j)} + X_{i_s j}^2/n$; 
                        $s \leftarrow s + 1$
                    }
                    $F'(v_{m+1}^{(k,j)}) \leftarrow F'(v_m^{(k,j)}) +
                    A^{(k,j)}(v_{m+1}^{(k,j)}-v_m^{(k,j)})$\;
                    \If{$F'(v_m^{(k,j)})<0\leq F'(v_{m+1}^{(k,j)})$}{
                        Set $\hat{\beta}_j^{(k,j)} \leftarrow v_m^{(k,j)}-\frac{F'(v_m^{(k,j)})}{A^{(k,j)}}$\;
                        \textbf{break}\;
                    }
                }
            }
            $r_i^{(k,j)} \leftarrow r_i^{(k,j-1)}- X_{ij}\left(\hat{\beta}_j^{(k,j)}-\hat{\beta}_j^{(k,j-1)}\right)$ for $i=1,\ldots,n$\;
        }
        $k \leftarrow k+1$\;
    }
    \tcp{Step 3: KKT violation check for $j \notin \mathcal{E}$}
    \For{$j \notin \mathcal{E}$}{
        Compute $r_{ij}^{(k,j-1)}$ and evaluate
        $|F'(0;\mathcal{I}^{(k,j-1)}(0))|$ using \eqref{e:slope}\;
        \If{$|F'(0;\mathcal{I}^{(k,j-1)}(0))|>\lambda_{\ell}\alpha$}{
            Add $j$ to $\mathcal{E}$\;
        }
    }
    \If{any $j$ added to $\mathcal{E}$}{
        Return to Step 2 and repeat until no violations remain\;
    }
    \tcp{Step 4: Update adaptive slope for next $\lambda$}
    Compute $c_j(\lambda_\ell)=-\frac{1}{n}\sum_{i=1}^nX_{ij}\rho_{\delta}'(r_i^{(k,p)})$ for $j=1,\ldots,p$, and set $M_\ell=\frac{\max_{j=1,\ldots,p} |c_j(\lambda_{\ell-1}) -
    c_j(\lambda_\ell)|}{\alpha(\lambda_{\ell-1} - \lambda_\ell)}$\;
    Store $\hat{\beta}(\lambda_\ell) \leftarrow \hat{\beta}^{(k,p)}$\;
}
\Return{$\{\hat{\beta}(\lambda_{\ell})\}_{\ell=1,\ldots,\ell_{\max}}$}\;
\end{algorithm}
}

\noindent We then generate responses according to $y_i=X_{i}^{\top}\beta+e_{i}$, $i=1,\ldots,n$, where $n$ is the number of samples and $e_i \stackrel{i.i.d.}{\sim}\mathcal{N}(0,1)$. We vary the sample sizes across $n\in\{100,500,1000\}$ and the dimensions across $p\in\{100,500,1000\}$. The columns of the design matrix $X_i$ are generated with covariance matrices $\Sigma_X$ according to the following scenarios:
\begin{enumerate}
    \item Compound $\mathcal{N}$: $X_i\sim\mathcal{N}(\boldsymbol{0}_{p},\Sigma_{X}=(\mathbf{1}_{\{i=j\}}+0.8\times\mathbf{1}_{\{i\neq j\}}))$.
    \item AR $t_2$: $X_i\sim t_{2}(\boldsymbol{0}_{p},\Sigma_{X}=(0.8^{|i-j|}))$.
    \item Contaminated AR: $X_i=[X_{1:p-1},X_{p}]^{\top}$, where $X_{1:p-1}\sim\mathcal{N}(\boldsymbol{0}_{p-1},\Sigma_{X}=(0.8^{|i-j|}))$, $X_{p}\sim t_{1}$, and $X_{1:p-1} \perp X_{p}$.
    \item Block AR $(\mathcal{N},t_1)$: $X_i=[X_{1:p/2},X_{p/2+1:p}]^{\top}$, where $X_{1:p/2}\sim t_1(\boldsymbol{0}_{p/2},\Sigma_{X}=(0.2^{|i-j|}))$, $X_{p/2+1:p}\sim\mathcal{N}(\boldsymbol{0}_{p/2},\Sigma_{X}=(0.8^{|i-j|}))$, and $X_{1:p/2} \perp X_{p/2+1:p}$.
\end{enumerate}
These scenarios are designed to evaluate the methods under a range of correlation structures and tail behaviors, including both heavy-tailed and mixed distributions. They allow us to assess the robustness of the proposed algorithms in settings that mimic practical high-dimensional data challenges. The simulation is repeated 100 times for each setting.

Figure \ref{fig:numerical_time} presents the runtime (measured in CPU seconds) across 100 independent replications. Note that the proposed methods are consistently faster than the alternative benchmarks across all combinations of $n$ and $p$, though their relative performance advantage varies by scenario. These results highlight the scalability of our approach and demonstrate its computational efficiency in high-dimensional regimes.

\begin{figure}[t]
\begin{center}
\includegraphics[width=1\textwidth,height=0.4\textheight]
 {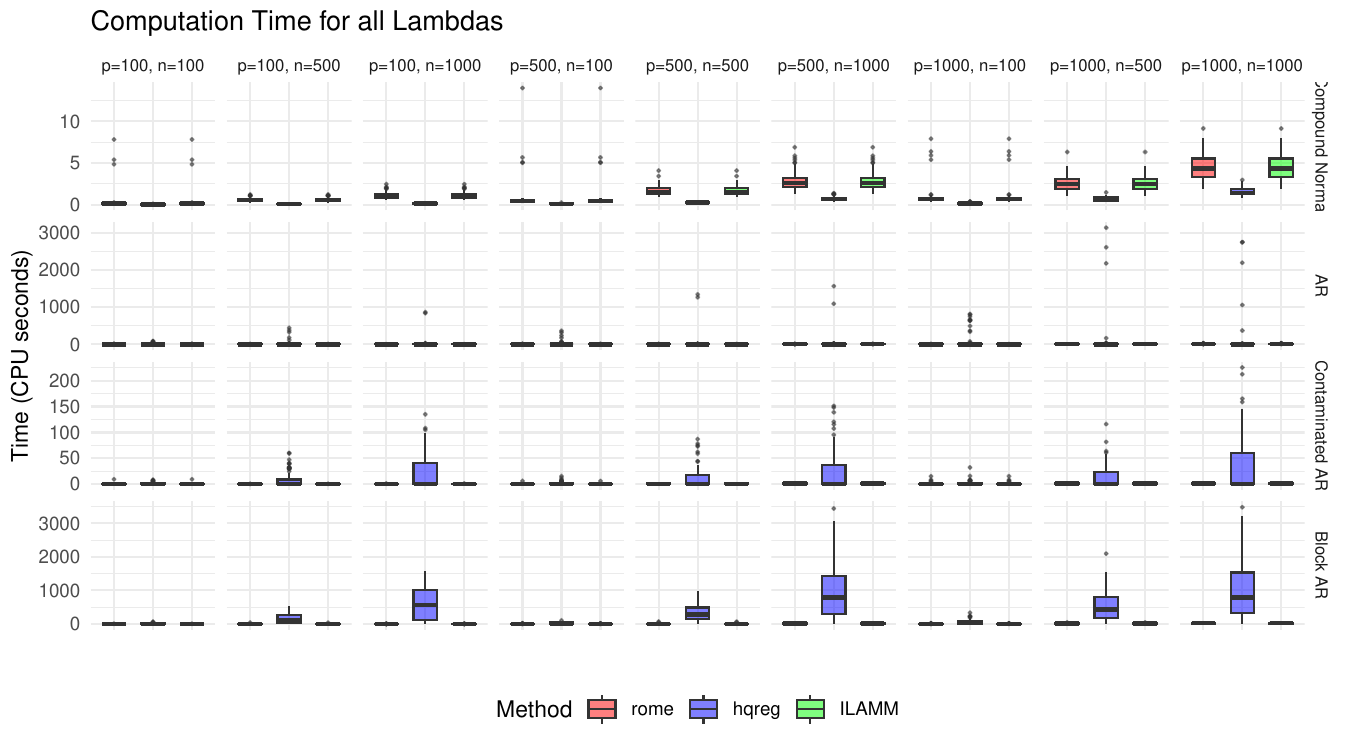}
 \end{center}
\caption{Boxplots of runtimes (CPU seconds) for the same 100 values of $\lambda$ across 100 replications. Each column corresponds to a different $(n,p)$ setting, and each row corresponds to one of four scenarios for $\{X_i\}$. The proposed method (\texttt{rome}) is shown in red, and the benchmarks (\texttt{hqreg}, \texttt{ILAMM}) are shown in blue and green, respectively.}
\label{fig:numerical_time}
\end{figure}

\begin{figure}[t]
\begin{center}
\includegraphics[width=1\textwidth,height=0.4\textheight]
{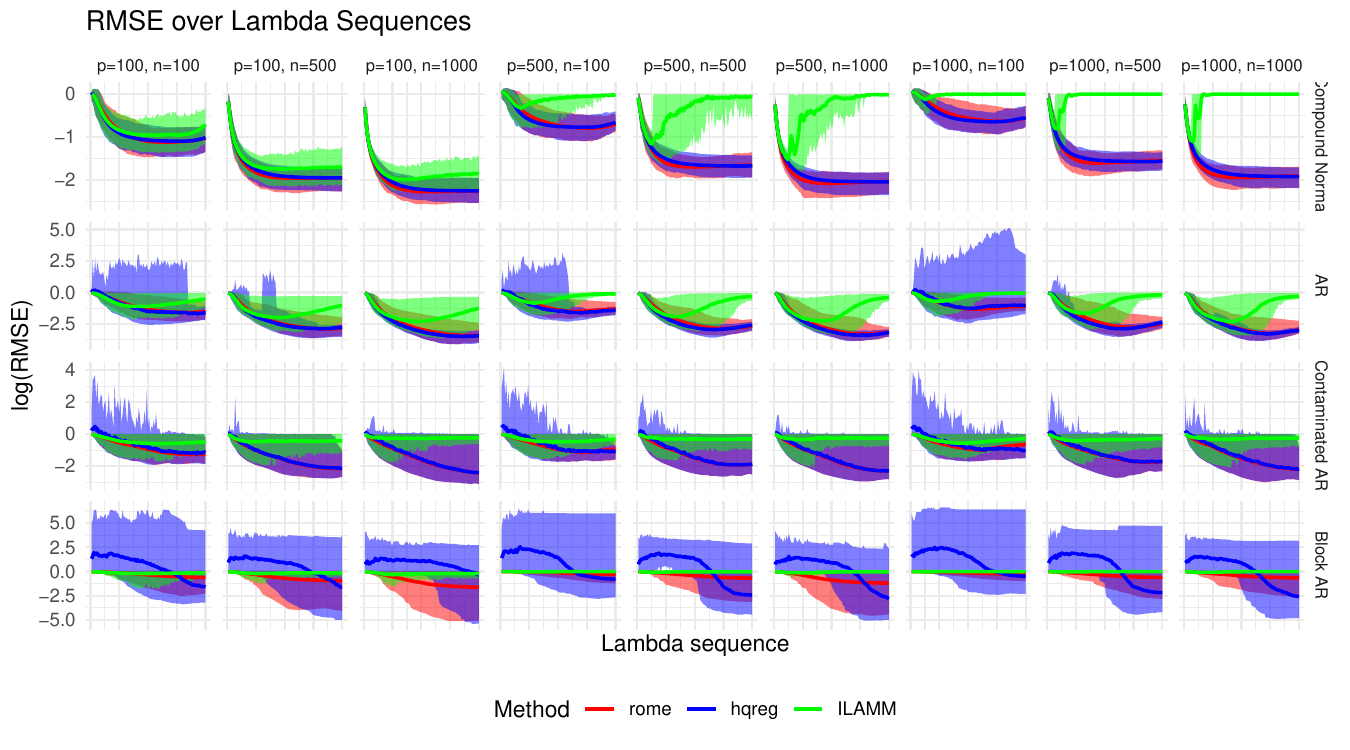}
\end{center}
\caption{Average normalized RMSE for the same 100 values of $\lambda$ across 100 replications. The mean RMSEs are shown as solid lines, and the shaded areas indicate the 5$^{\text{th}}$ and 95$^{\text{th}}$ percentile ranges of RMSEs. Each column corresponds to a different $(n,p)$ setting, and each row corresponds to one of four scenarios of $\{X_i\}$. The proposed method (\texttt{rome}) is shown in red, and the benchmarks (\texttt{hqreg}, \texttt{ILAMM}) are shown in blue and green, respectively.}
\label{fig:numerical_rmse}
\end{figure}

Figure \ref{fig:numerical_rmse} presents the log-scaled average normalized RMSE. The thick solid lines represent the averages over 100 independent replications at each grid point $\lambda_{\ell}$, with the shaded regions denoting the 5${^\text{th}}$ and 95$^{\text{th}}$ percentile bands. Across all combinations of $n$ and $p$, the RMSEs of the proposed method consistently decrease along the $\lambda$ sequence. In contrast, the alternative benchmarks struggle, exhibiting pronounced V-shaped curves in their averages. Notably, the \texttt{hqreg} method yields significantly higher RMSEs compared to both our approach and \texttt{ILAMM}. While its performance improves for sufficiently large $n$, it fails to provide reliable results across all high-dimensional cases. Given the data-generating process, these results demonstrate that our method is the most robust and reliable option under these challenging scenarios, maintaining stable performance across varying sample sizes and dimensions.

\subsection{Comparison of Convergence Rates}
\label{sse:experiment_convergence}

In this section, we present numerical comparisons of convergence rates for GD, CGD, and the proposed CD method to provide empirical support for the theoretical analysis in Proposition \ref{prop:convergence} (with the underlying derivations and details provided in the Supplementary Material). As publicly available code implementing the GD and CGD algorithms for regularized Huber regression does not exist, we implemented these routines in C and wrapped them in an R interface. To ensure a fair comparison, the warm-start and screening rules described in Sections \ref{sse:warm} and \ref{sse:screen} for the CD method are not applied. Furthermore, since CGD is a hypothetical method constructed specifically for the proof of Proposition \ref{prop:convergence}, the corresponding figure is placed in the Supplementary Material.

We again consider the linear regression model $y_i = X_{i}^{\top}\beta + e_{i}$, $i=1,\ldots,n$, where $e_{i}\stackrel{i.i.d.}{\sim}\mathcal{N}(0,1)$. The rows of the design matrix $X \in \mathbb{R}^{n \times p}$ are generated according to the following two design matrices:
\begin{enumerate}
    \item[5.] AR $\boldsymbol{t}_4$: $X_i\sim t_4(\boldsymbol{0}_{p},\Sigma_{X}=(\rho^{|i-j|}))$.
    \item[6.] Block AR $(\boldsymbol{\mathcal{N},\mathcal{N}})$: $X_i=[X_{1:p/2},X_{p/2+1:p}]^{\top}$, where $X_{1:p/2}\sim \mathcal{N}(\boldsymbol{0}_{p/2},\Sigma_{X}=(\rho_1^{|i-j|}))$, $X_{p/2+1:p}\sim\mathcal{N}(\boldsymbol{0}_{p/2},\Sigma_{X}=(\rho_2^{|i-j|}))$, and $X_{1:p/2} \perp X_{p/2+1:p}$.
\end{enumerate}
The correlation parameters are fixed at $\rho=0.4$ and $(\rho_1,\rho_2)=(0.2,0.8)$, respectively. Following the simulation framework established in \cite{yi2017semismooth}, the true coefficient vector $\beta$ is defined as
\begin{equation}\label{e:true_coefficients}
    \beta_j = (-1)^j\exp(-(j-1)/10),\ j=1,\ldots,p,
\end{equation}
We consider the dimension and sample sizes $(p,n) \in \{(100,500), (1000,100), (1000,500)\}$ and fix $\delta=0.5$. The regularization parameter is set to $\lambda = \epsilon\lambda_0$, where $\lambda_0$ is computed via \eqref{e:lambda_max} with $\epsilon=0.001$. We track the objective function value using the solution obtained over 50 iterations $(k)$. As a single update in either CGD or CD modifies only a single coordinate of $\beta$, one full iteration is defined as $p$ consecutive coordinate updates for both algorithms.

Figure \ref{fig:numerical_convergence} in the Supplementary Material illustrates the objective values as a function of $k$. The results demonstrate that the proposed coordinate descent (CD) method achieves a significantly faster convergence rate than both GD and CGD across all evaluated combinations of dimensions and design matrices, echoing the findings of \cite{saha2013nonasymptotic} for the Lasso penalty in least-squares settings. While the standard GD algorithm shows negligible progress and CGD exhibits only moderate reduction, the proposed algorithm rapidly decreases the objective function value, reaching high precision within the first few iterations.

\subsection{Performance on Checking Optimality Conditions}
\label{sse:experiment_kkt}

We investigate the effectiveness of checking the KKT conditions described in Section \ref{sse:kkt}. In these simulations, no screening rules are applied. We compare the runtimes between (i) applying KKT condition checks (Applied) and (ii) not applying them (N/A), similar to \cite{kim2025pathwise}. 

We employ the data-generating processes 5 and 6 outlined in Section \ref{sse:experiment_convergence} and the $p$-dimensional true coefficient vector $\beta$ in \eqref{e:true_coefficients}. We vary the dimensions across $p\in\{50, 100,200,500,1000\}$ and the sample sizes across $n\in\{100,500\}$, respectively. We fix $\delta=0.5$. The reported results represent the average CPU time (in seconds) over 100 replications. For each replication, the runtime is measured by computing the optimization solutions using warm starts along a 100 $\lambda$-sequence constructed as described in Section \ref{sse:warm}.

Table \ref{tab:numerical_kkt} summarizes the computational results. Incorporating the KKT condition checks yields substantial performance gains over the baseline approach across all combinations of $n$ and $p$, as well as across both data-generating scenarios. Interestingly, the computational runtime scales predictably with respect to both the sample size and feature dimension under both configurations. While the growth in runtime remains modest relative to $p$, the more pronounced increase with respect to $n$ is driven by the internal sorting operations required by the Huber loss optimization routine.

\begin{table}[t]
\centering
\resizebox{0.75\columnwidth}{!}{%
\begin{tabular}{cccccccc}
\cline{4-8}
&  & KKT check & $p=50$ & $p=100$ & $p=200$ & $p=500$ & $p=1000$ \\ \cline{1-8}
\multirow{4}{*}{AR} 
& \multirow{2}{*}{$n=100$} & Applied & 0.270 & 0.593 & 1.188 & 1.942 & 1.733 \\ \cline{3-3}
&  & N/A & 0.829 & 1.691 & 3.523 & 7.723 & 13.585 \\ \cline{2-8}
& \multirow{2}{*}{$n=500$} & Applied & 0.930 & 1.380 & 2.605 & 7.959 & 16.690 \\ \cline{3-3}
&  & N/A & 4.138 & 9.643 & 20.241 & 52.171 & 105.090 \\ \cline{1-8}
\multirow{4}{*}{Block AR} 
& \multirow{2}{*}{$n=100$} & Applied & 0.269 & 0.578 & 1.117 & 1.298 & 1.176 \\ \cline{3-3}
&  & N/A & 0.862 & 1.636 & 3.492 & 6.804 & 12.650 \\ \cline{2-8}
& \multirow{2}{*}{$n=500$} & Applied & 0.806 & 1.307 & 2.475 & 7.168 & 14.566 \\ \cline{3-3}
&  & N/A & 4.668 & 10.151 & 20.671 & 52.210 & 104.116 \\ \cline{1-8}
\end{tabular}%
}
\caption{Average runtimes (CPU seconds) for the same 100 values of $\lambda$ across 100 replications. The cases where KKT checking is applied (Applied) and not applied (N/A) are compared.}
\label{tab:numerical_kkt}
\end{table}

\subsection{Performance on Strong Rules}
\label{sse:experiment_rules}

We validate the computational screening strategy described in Section \ref{sse:screen}, implemented in the \texttt{rome} package. Specifically, we assess the correctness of the selection rules by comparing (i) the proposed adaptive sequential strong rule (ASR) against (ii) the standard sequential strong rule (SSR), where the slopes in \eqref{e:adaptive_slope} are fixed at $M_{\ell} = 1$ for all $\ell$. Following the simulation framework of \cite{tibshirani2012strong}, we evaluate whether ASR or SSR incurs any screening violations. Along the regularization path $\{\lambda_{\ell}\}_{\ell=0,1,\ldots,100}$, a violation is defined as a predictor that violates the screening threshold but possesses a nonzero coefficient in the final optimization solution.

We again use the data-generating processes 5 and 6 in Section \ref{sse:experiment_convergence}. For the $p$-dimensional true coefficient vector $\beta$, $10\%$ of the entries are randomly selected to be 1, while the remaining entries are set to 0. We set $n=100$, vary the dimensions across $p\in\{100,500,1000\}$, and use $\delta=0.5$.  Each simulation setup is replicated 100 times. For each replication, the model is evaluated along a 100 $\lambda$-sequence constructed as described in Section \ref{sse:warm}.

Figure \ref{fig:numerical_violation} displays the result of violation rates. While the number of violations under the SSR framework tends to escalate as the dimension $p$ grows, particularly at smaller $\lambda$ values, the ASR framework consistently maintains nearly zero violations across the entire regularization path, regardless of the problem dimension. Interestingly, the proportion of SSR violations depends heavily on the specific $\lambda$ threshold and the feature dimension, whereas it remains relatively invariant across the different correlation scenarios. These empirical findings demonstrate that ASR provides a substantially more reliable and stable screening mechanism than SSR, particularly in high-dimensional settings.

In addition to the settings considered above, we also investigate the performance of the proposed algorithm in sparse vector autoregressive (VAR) estimation under heavy-tailed innovations. Since each row of a VAR coefficient matrix can be estimated by solving a penalized regression problem with lagged observations as predictors, persistent autoregressive dynamics naturally produce highly correlated design matrices, providing a realistic benchmark for our motivating scenario. The simulation results are reported in the Supplementary Material.

\begin{figure}[t]
\begin{center}
\includegraphics[width=1\textwidth,height=0.3\textheight]
{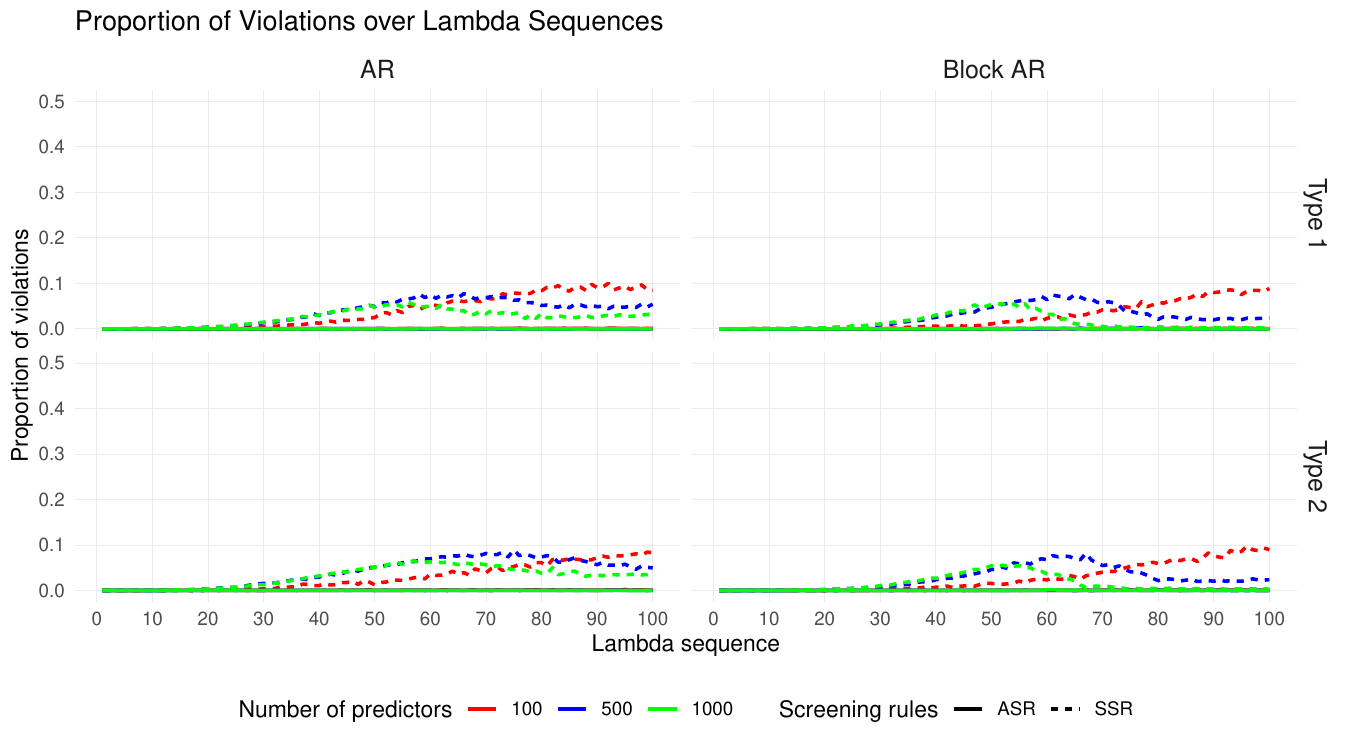}
\end{center}
\caption{Average proportion of violations among the total number of variables across 100 replications under different rules (ASR and SSR) with \texttt{rome}. The columns correspond to scenario 1 (AR) and scenario 2 (Block AR), while the rows present the results for type 1 and type 2. Within each panel, six cases are displayed; solid and dotted lines represent ASR and SSR, respectively, and the colors red, blue, and green correspond to $p=100$, $500$, and $1000$, respectively.}
\label{fig:numerical_violation}
\end{figure}

\subsection{Real Data Application}
\label{sse:realdata}

Finally, we present an application in which the characteristics of the dataset resemble the challenging scenarios considered in the simulation studies. We analyze a dataset from electron-probe X-ray microanalysis of (n=180) archaeological glass vessels \citep{janssens1998composition}, where each vessel is represented by a spectrum across 1,920 frequencies, along with the recorded contents of 13 chemical compounds, including lead oxide (PbO). The goal is to predict PbO content using spectral frequencies as predictors. Since X-ray spectra can be acquired non-invasively, accurate prediction is particularly important: it enables researchers to infer the chemical composition of archaeological artifacts without damaging them. Since values below 15 and above 500 are nearly zero with little variability, we retain only frequencies 15-500, resulting in $p=486$ predictors. This dataset has been used in various studies of high-dimensional robust regression \cite[e.g.,][]{maronna2011robust,smucler2017robust,loh2021scale}.

As seen in the left panel of Figure \ref{fig:real_data}, the spectral predictors form distinct blocks with sharp, large-magnitude shifts in scale across frequencies. This structure is reflected in the sample covariance matrix (right panel of Figure \ref{fig:real_data}), where tightly correlated frequency clusters appear as blocks along the diagonal, and several off-diagonal blocks show strong negative correlation. Combined with the heavy-tailed marginal distributions typical of spectral intensity data, this pronounced block correlation structure closely mirrors the challenging conditions discussed in Section \ref{sse:related}: highly correlated predictors and heavy-tailed noise, making this dataset a real-data test of the proposed method.

\begin{figure}[t]
\centering
\begin{subfigure}[b]{0.45\textwidth}
\includegraphics[width=1\textwidth,height=0.2\textheight]{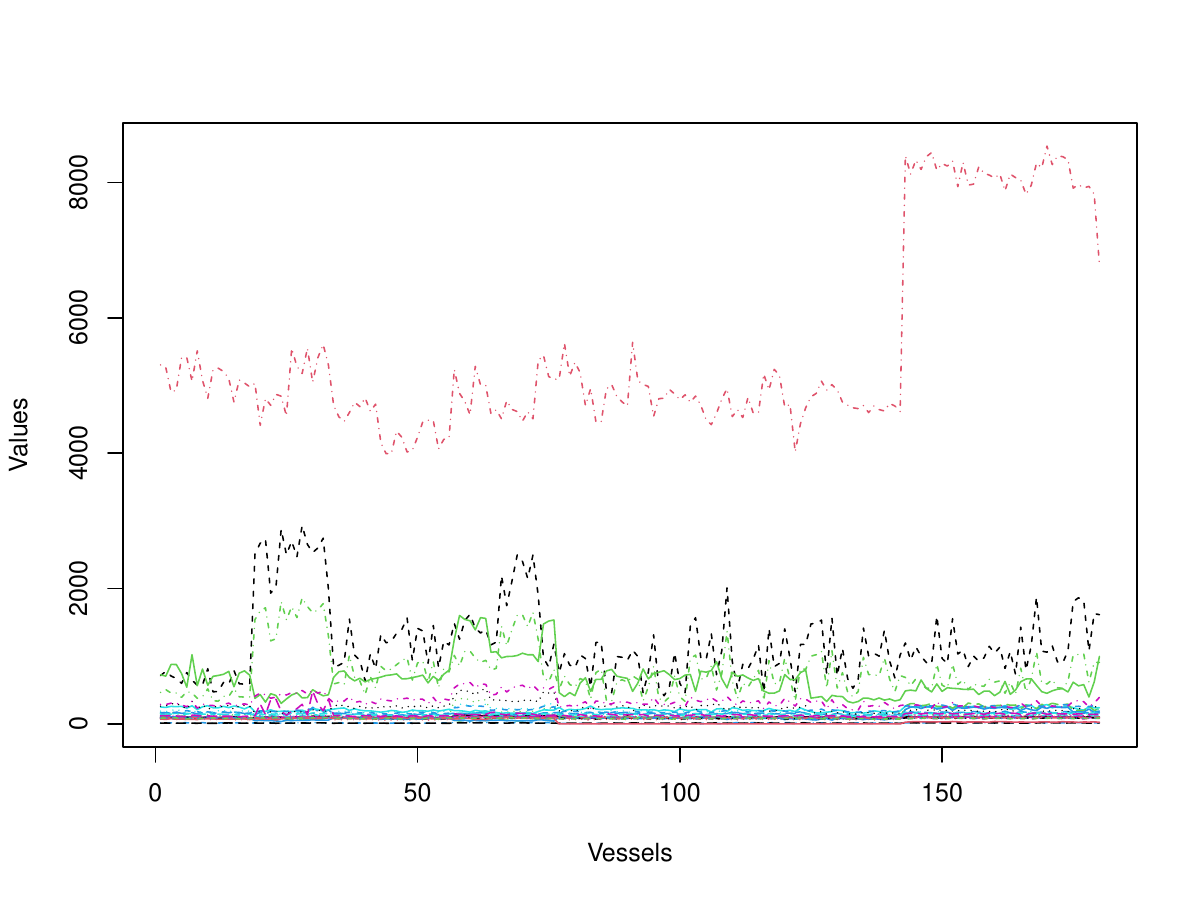}
\end{subfigure}
\begin{subfigure}[b]{0.45\textwidth}
\includegraphics[width=1\textwidth,height=0.2\textheight]{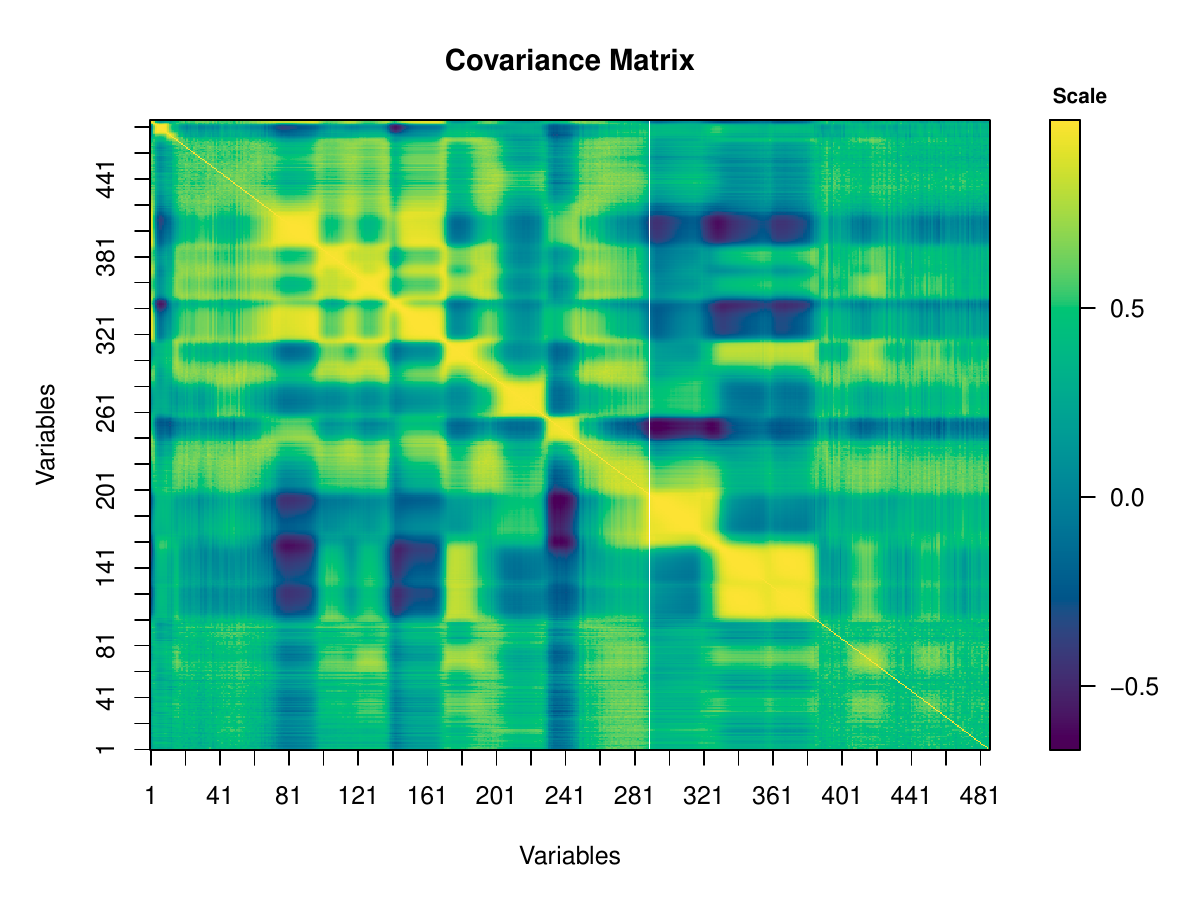}
\end{subfigure}
\caption{Plots of the X-ray microanalysis over 30 randomly chosen frequencies (left panel) and the sample covariance matrix of the covariates (right panel).}
\label{fig:real_data}
\end{figure}

First, we illustrate cross-validation results, similar to those shown in \cite{wu2008coordinate}. For a fixed $\alpha=0.3$, we select the optimal $\lambda$ via cross-validation by minimizing (i) the null deviance, (ii) the mean absolute error, and (iii) the root mean squared error. The results are displayed in Figure \ref{fig:real_cv} in the Supplementary Material. Note that the three metrics do not show substantial differences, and all exhibit a weak U-shape across the $\lambda$ values. The two thick lines indicate the $\lambda$ values chosen: one minimizes the metric, and the other represents the largest $\lambda$ for which the cross-validated error is within one standard error from the minimum. This approach is widely used when interpreting cross-validation results \cite[e.g.,][]{friedman2010regularization}.

Next, we report the number of selected nonzero coefficients based on the full dataset, alongside the model's predictive performance, following the validation frameworks established in existing penalized robust regression literature \cite[e.g.,][]{peng2015iterative,gu2018admm}. Specifically, to evaluate predictive performance, we generate 50 random partitions of the data. Each partition splits the dataset into a training set $I_{\textrm{train}}$ containing 120 glass vessels and a validation set $I_{\textrm{val}}$ containing 60 glass vessels. For each partition, the model parameters are estimated using 5-fold cross-validation on $I_{\textrm{train}}$, utilizing the null deviance as the optimal tuning criterion. The predictive performance is then evaluated on the validation set by computing the average prediction error
\begin{displaymath}
    \frac{1}{|I_{\textrm{val}}|} \sum_{i \in I_{\textrm{val}}} \left(y_i - X_i^{\top}\hat{\beta}(I_{\textrm{train}})\right)^2,
\end{displaymath}
where $|I_{\textrm{val}}| = 60$. We evaluate this performance across three distinct robustness thresholds $\delta \in \{0.5, 1.0, 1.5\}$. To highlight the advantage of employing a robust framework, we benchmark our results against regularized least squares, implemented via the \texttt{glmnet} package. To ensure a fair comparison, the same $\alpha$ is enforced, while the regularization parameter $\lambda$ for \texttt{glmnet} is selected independently using the package's cross-validation routine.

The results are displayed in Table \ref{tab:real_selection}. Among the choices of $\delta$ for the \texttt{rome} model, we observe a consistent monotonic trend: as $\delta$ increases, the number of nonzero coefficients selected decreases for both the full dataset and the random partitions. For the random partitions, the average number of selected variables is generally higher than that obtained from the full dataset, reflecting the structural variability introduced across different training subsets. In terms of predictive performance, the validation errors increase as the $\delta$ threshold grows, with $\delta=0.5$ achieving the best overall accuracy. Note that, while \texttt{glmnet} selects a significantly sparser model than \texttt{rome} both on the full data and across partitions, it yields a higher prediction error than the \texttt{rome} formulation at lower $\delta$ thresholds ($\delta = 0.5,1.0$). These results demonstrate that robust regression operates as desired, suggesting an effective trade-off between model sparsity and predictive precision. We are able to obtain a similar result with a different choice of $\alpha$.

\begin{table}[t]
\centering
\resizebox{0.75\columnwidth}{!}{%
\begin{tabular}{cc c cc}
 & \multicolumn{1}{c}{All Data} & & \multicolumn{2}{c}{Random Partition} \\ 
\cline{2-2} \cline{4-5} \rule{0pt}{3ex}%
Model & \# of nonzeros & & Average \# of nonzeros & Prediction Error \\ \hline
$\delta=0.5,\alpha=0.3$ & 193 & & 262.36 (92.535) & 0.025 (0.010) \\ \cline{1-1}
$\delta=1.0,\alpha=0.3$ & 185 & & 244.36 (103.774) & 0.029 (0.017) \\ \cline{1-1}
$\delta=1.5,\alpha=0.3$ & 164 & & 190.96 (73.081) & 0.059 (0.029) \\ \cline{1-1}
$\texttt{glmnet},\alpha=0.3$ & 38  & & 46.54 (8.721)   & 0.031 (0.016) \\ \hline
\end{tabular}%
}
\caption{Analysis of the X-ray microanalysis glass vessels data reported in \cite{janssens1998composition} by \texttt{rome} with three different values of $\delta$ and \texttt{glmnet} with the same $\alpha$. For the random partition, the mean and standard deviation (in parentheses) for the 50 replications are reported.}
\label{tab:real_selection}
\end{table}

\section{Discussion}
\label{se:discussion}

The proposed coordinate descent framework achieves superior performance relative to state-of-the-art benchmark methods under highly challenging statistical regimes. Specifically, in the presence of ill-conditioned Hessians driven by highly correlated covariates and heavy-tailed error distributions, the proposed method yields advantages in computational time and estimation accuracy, demonstrating stability and efficiency. Theoretical convergence guarantees support the empirical performance gains observed alongside computational acceleration strategies, including the optimality verification, which is also theoretically established. Extensive simulation studies and a real-world data application further demonstrate the practical applicability of the proposed framework in these settings.

Despite these advances, several directions remain for future research. First, the proposed method cannot be directly extended to nonconvex robust loss functions, such as the Cauchy or Tukey biweight losses \citep[e.g.,][]{loh2017statistical}, since their interval-dependent slopes make index tracking prohibitively complex. Developing effective approximations for these loss functions is therefore an important direction. Second, although the adaptive screening rule for penalized Huber regression consistently prevents violations in variable selection in our empirical studies, a deeper theoretical investigation of its properties is warranted. Finally, analogous to the coordinate descent algorithm for the scaled Lasso \citep{sadhukhan2025autotune}, extending the proposed framework to jointly estimate the location and scale parameters in regularized robust regression is another promising direction.

\section*{Acknowledgments}

SB acknowledges partial support from NSF CAREER award DMS-2239102, NSF awards DMS-1812128, DMS-2210675, and NIH awards R01GM135926, R21NS120227.

\section*{Data Availability Statement}

The R package \texttt{rome} and the code used for the simulation studies and the data application in Section \ref{se:numerical} are available on GitHub at \href{https://github.com/yk748/rome}{https://github.com/yk748/rome}. The dataset used in Section \ref{sse:realdata} is also provided in the repository.

{\small
\bibliographystyle{apalike}
\bibliography{ref}
}


\clearpage
\appendix

\setcounter{section}{0}
\setcounter{page}{1}
\setcounter{figure}{0}
\setcounter{equation}{0}
\renewcommand{\thesection}{S.\arabic{section}}
\renewcommand{\thepage}{\arabic{page}}
\renewcommand{\thetable}{S.\arabic{table}}
\renewcommand{\thefigure}{S.\arabic{figure}}
\renewcommand{\theequation}{S.\arabic{equation}}

\begin{center}
\textbf{\LARGE {Supplementary Material for \\ ``Exact Coordinate Descent for High-Dimensional Regularized Huber Regression}''}
\end{center}    

\begin{center}
Younghoon Kim$^{*1}$, Po-Ling Loh$^2$, and Sumanta Basu$^{3}$
\end{center}

\begin{center}
$^1$Department of Mathematics, University of Alabama \\
$^2$Department of Pure Mathematics and Mathematical Statistics, University of Cambridge
$^3$Department of Statistics and Data Science, Cornell University
\end{center}

\def\thefootnote{$*$}\footnotetext{Corresponding author. Email: ykim124@ua.edu}

Section \ref{ap:figures} presents additional figures, including convergence comparisons and cross-validation results. Section \ref{ap:simulations} describes additional simulation studies for application to high-dimensional time series models and reports further performance comparisons. The proofs of the main theoretical results are provided in Sections \ref{ap:proofs_convergence} and \ref{ap:proofs_kkt}.

\section{Additional Figures}
\label{ap:figures}

\begin{figure}[h]
\centering
\includegraphics[width=0.95\textwidth]{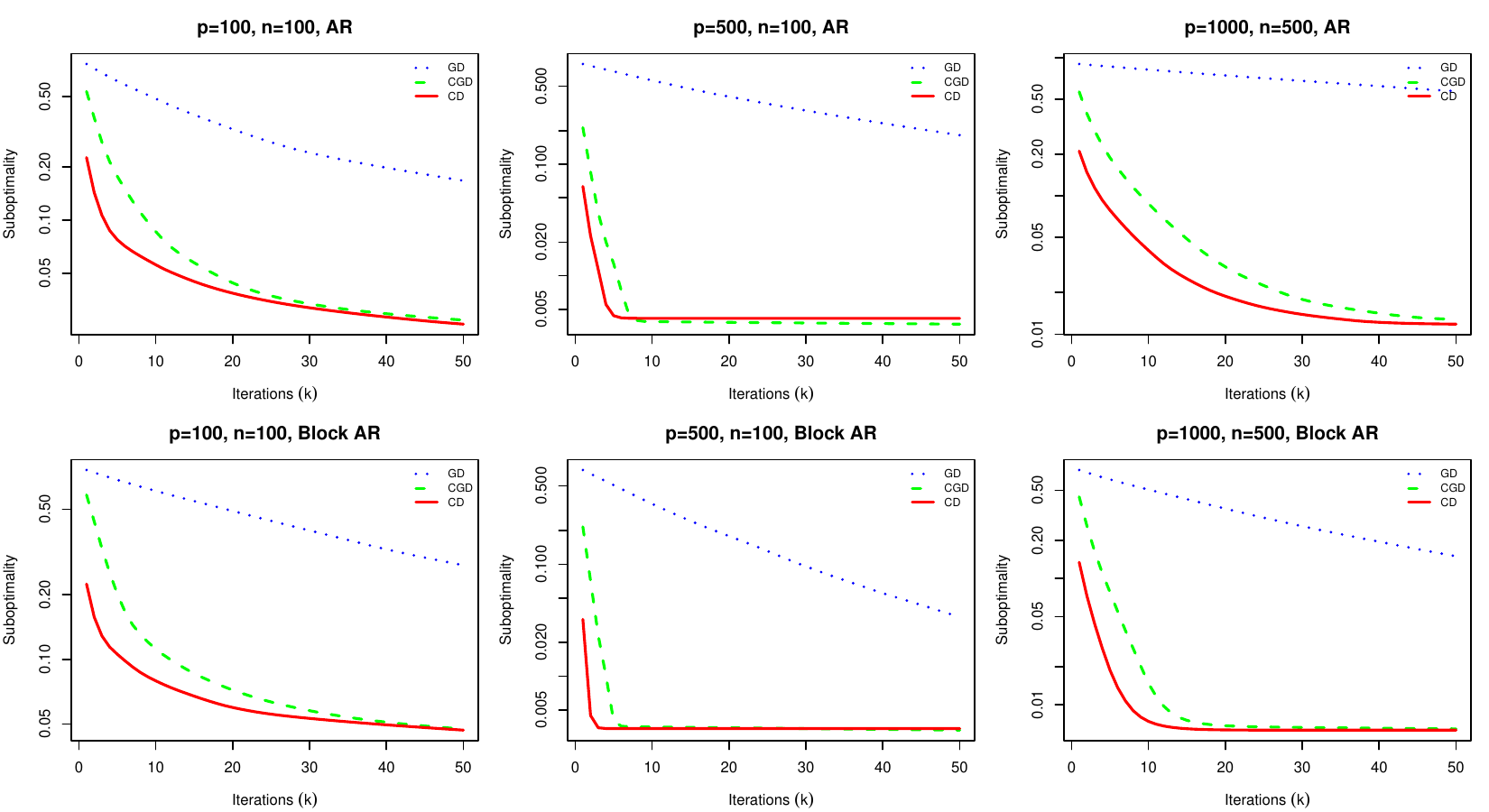}
\caption{Convergence speed, represented by the values of the objective function of the gradient descent (GD), coordinate gradient descent (CGD), and proposed coordinate descent (CD) methods over iterations $(k)$, represented by dotted blue, dashed green, and solid red lines, respectively.}
\label{fig:numerical_convergence}
\end{figure}

\begin{figure}[h]
\centering
\begin{subfigure}[b]{0.3\textwidth}
\includegraphics[width=1\textwidth,height=0.2\textheight]{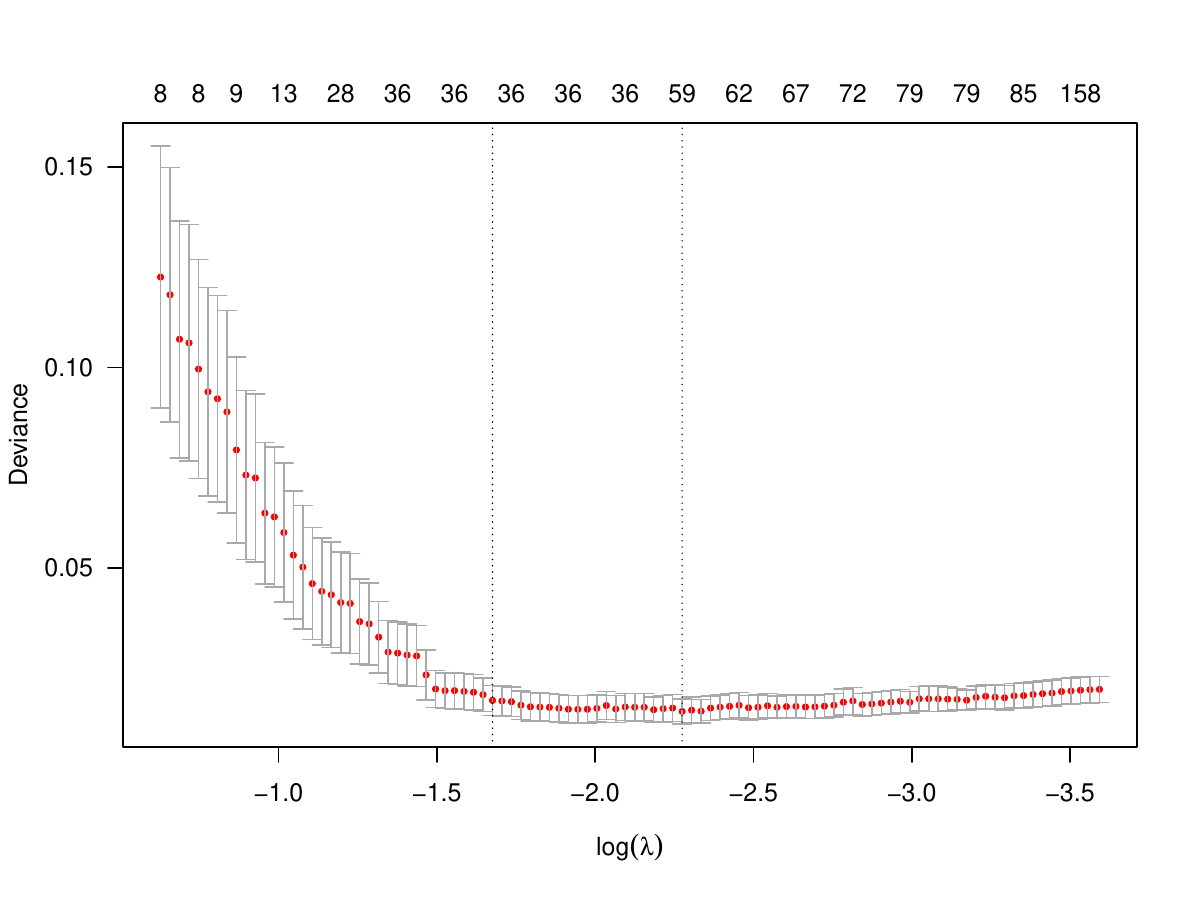}
\end{subfigure}
\begin{subfigure}[b]{0.3\textwidth}
\includegraphics[width=1\textwidth,height=0.2\textheight]{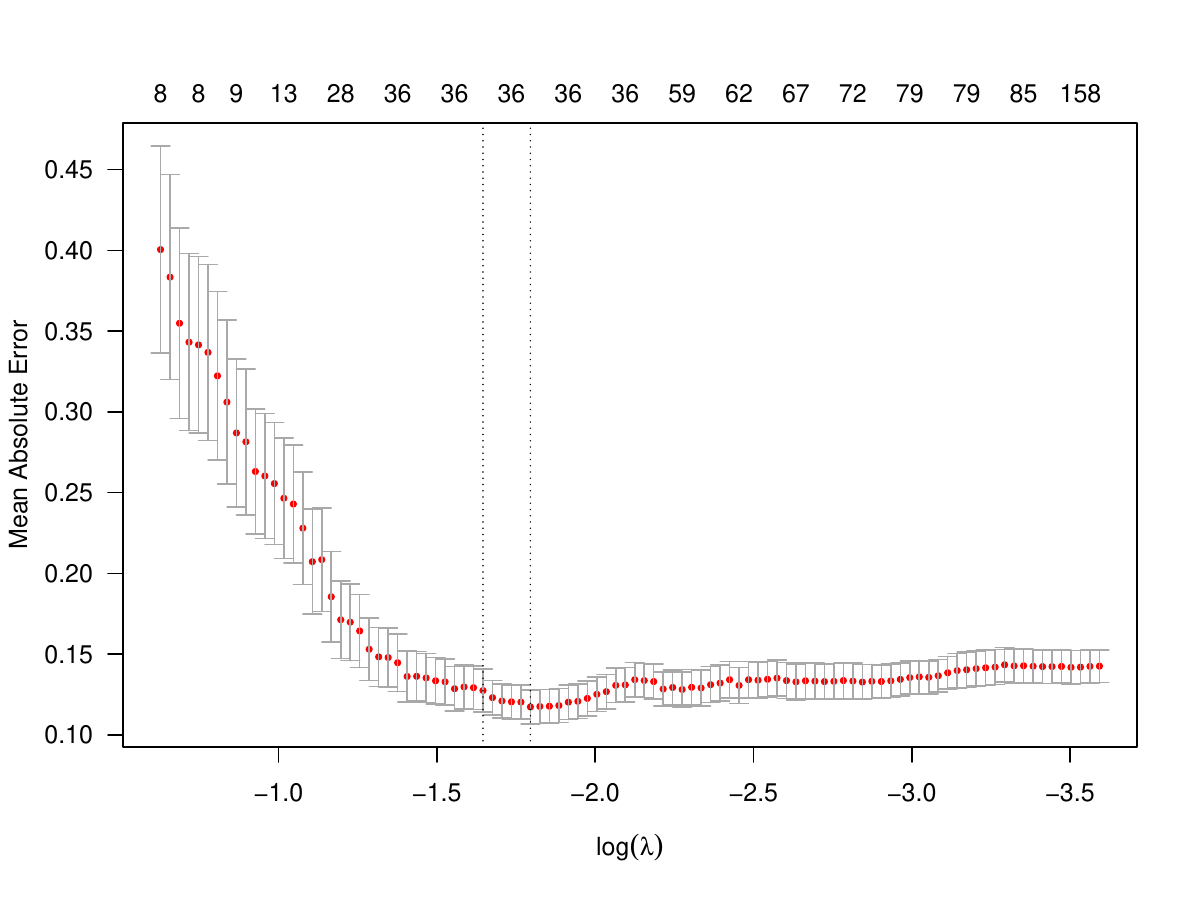}
\end{subfigure}
\begin{subfigure}[b]{0.3\textwidth}
\includegraphics[width=1\textwidth,height=0.2\textheight]{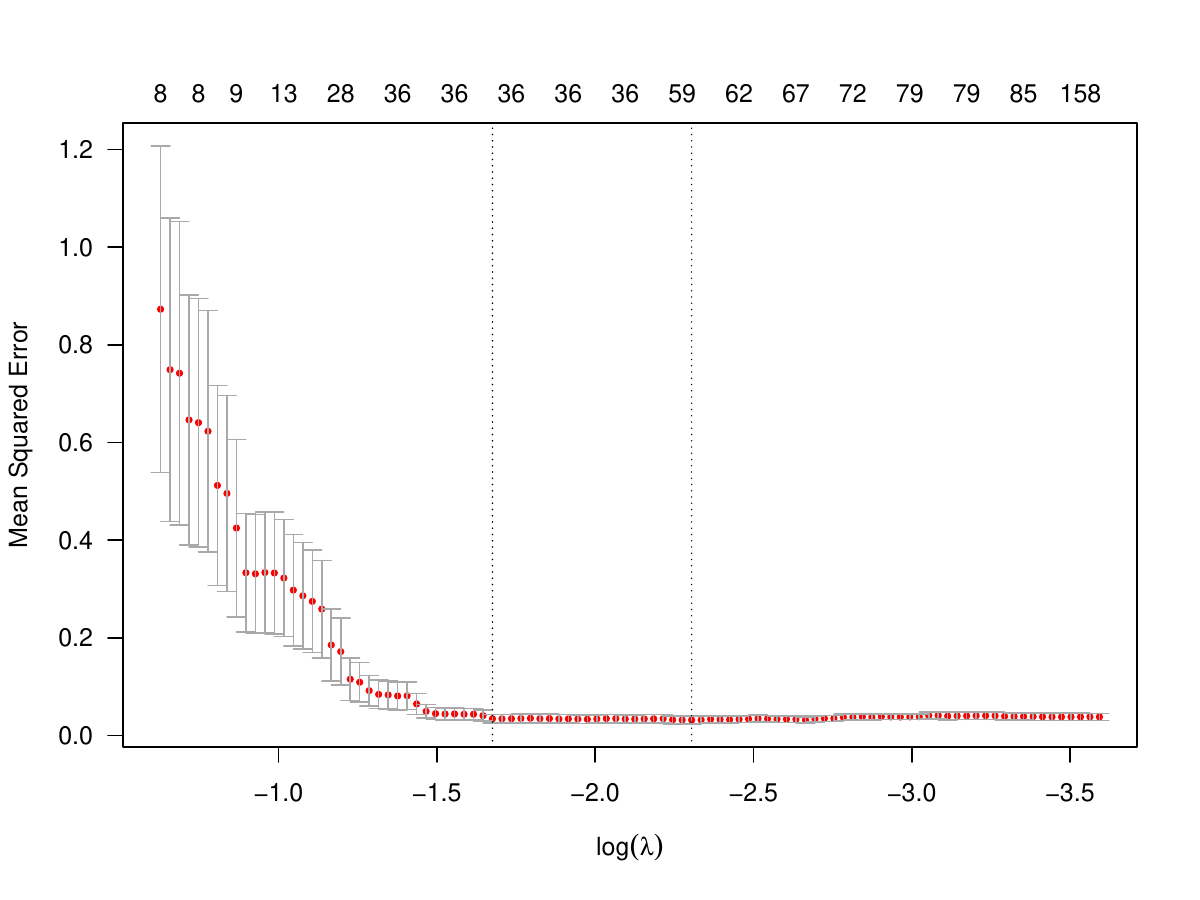}
\end{subfigure}
\caption{Plot of 10-fold cross-validation error and number of predictors versus 100 values of $\lambda$ for $\delta=0.5$ and $\alpha=0.3$. The evaluation criteria used are null deviance (left panel), mean absolute error (middle panel), and root mean square error (right panel), respectively.}
\label{fig:real_cv}
\end{figure}

\section{Additional Simulations}
\label{ap:simulations}

To further evaluate the proposed algorithm in a realistic correlated-design setting, we consider sparse VAR models under heavy-tailed innovations. Specifically, observations are generated from a $p$-dimensional VAR(1) model
\begin{displaymath}
    Y_t = \Phi Y_{t-1}+\varepsilon_t,
\end{displaymath}
where the innovation process $\{\varepsilon_t\}$ is generated according to one of the cases described below. $n=100$ is fixed, and $p\in\{20,50,100,200\}$ are considered. The $i^{\textrm{th}}$ component of the VAR model can be written as
\begin{displaymath}
    Y_{it}=\sum_{j=1}^{p}\phi_{ij}Y_{j,t-1}+\varepsilon_{jt},
\quad i=1,\ldots,p.   
\end{displaymath}
Thus, each row of $\Phi$ is estimated separately by solving a regularized Huber regression with the lagged vector $Y_{t-1}$ as predictors. Similar to studies of the high-dimensional VAR modeling \cite[e.g.,][]{qiu2015robust_supp,han2015direct_supp}, the transition matrix $\Phi\in\mathbb{R}^{p\times p}$ is constructed according to three different dependence structures:
\begin{enumerate}
    \item Banded structure, where nonzero coefficients are concentrated around the main diagonal, i.e., 
    $\Phi_{ij}=0$ if $|i-j|>2$.
    \item Hub structure, where $5\%$ of variables are selected as hubs and connected with all other variables. Specifically, the corresponding rows and columns of the transition matrix contain nonzero coefficients, while the remaining off-diagonal entries are set to zero.
    \item Random structure, where $10\%$ of off-diagonal entries are selected uniformly at random and assigned nonzero coefficients.
\end{enumerate}
In all settings, the diagonal entries of $\Phi$ are initialized as $0.7$, and the nonzero off-diagonal entries are sampled from a uniform distribution on $(0,0.3)$. The transition matrix is rescaled whenever necessary to ensure stationarity with spectral radius of $\Phi$ below $0.95$. A burn-in period of 500 observations is discarded before estimation.

Motivated by heavy-tailed high-dimensional time series modeling \citep{wong2020lasso_supp}, we consider two innovation mechanisms motivated by heavy-tailed high-dimensional time series modeling \cite[e.g.,][]{wong2020lasso_supp}: (i) symmetric sub-Weibull innovations with Weibull shape parameter $1/2$ and scale parameter 1 (Sub-Weibull); and 
(ii) a clipped state-dependent heteroscedastic innovation model (Clipped-ARCH),
\begin{displaymath}
    Y_t = \Phi Y_{t-1} + \textrm{clip}_{[-1,1]}\left(\|Y_{t-1}\|_2^2\right)\varepsilon_t,
\end{displaymath}
where the components of $\varepsilon_t$ are independently generated from the same symmetric sub-Weibull distribution as in (i), and the clipping function is defined as 
\begin{displaymath}
    \textrm{clip}_{[a,b]}(x)=\min\{\max\{x,a\},b\}.
\end{displaymath}
The clipping operation is used to limit the impact of extreme past observations, inducing heteroscedasticity with heavier-tailed behavior than Gaussian innovations while maintaining stability of the high-dimensional VAR process.

For each simulated data set, each row of the transition matrix is estimated separately by solving a regularized Huber regression with the lagged observations serving as predictors. As in Section \ref{sse:experiment_benchmark}, we compare the proposed algorithm (\texttt{rome}) against \texttt{hqreg} and \texttt{ILAMM} using the same regularization path for each $i^{\textrm{th}}$ variable. Computational efficiency is measured by the log-scaled time required to compute the full solution path, and log-scaled estimation accuracy is assessed by the relative Frobenius norm error along the entire solution path. The simulation is repeated 100 times for each setting.

Figure \ref{fig:numerical_rmse_time} represents the computational time (top panel) and estimation accuracy (bottom panel), respectively. In terms of computation time, \texttt{hqreg} is competitive in lower-dimensional settings. However, this advantage gradually diminishes as the dimension increases. Moreover, the runtime behavior varies substantially across the dependence structures. The computation times of the proposed method and \texttt{ILAMM} increase gradually with the dimension, while \texttt{hqreg} becomes increasingly unstable, especially under the Hub and Random structures. Throughout all settings, the RMSE profiles of \texttt{hqreg} begin to diverge at larger values of the regularization parameter, with the discrepancy progressively propagating along the remainder of the solution path.

Notably, the proposed method and \texttt{ILAMM} exhibit comparable computational performance across all settings, with negligible differences in both the average computation time and its variability. In terms of estimation accuracy, the two methods produce similar RMSEs along the entire regularization path, while the proposed method consistently exhibits lower variability. These results further demonstrate the robustness of the proposed algorithm for sparse VAR estimation with highly correlated lagged predictors and heavy-tailed innovations.

\begin{figure}[t]
\begin{center}
\includegraphics[width=1\textwidth,height=0.3\textheight]
{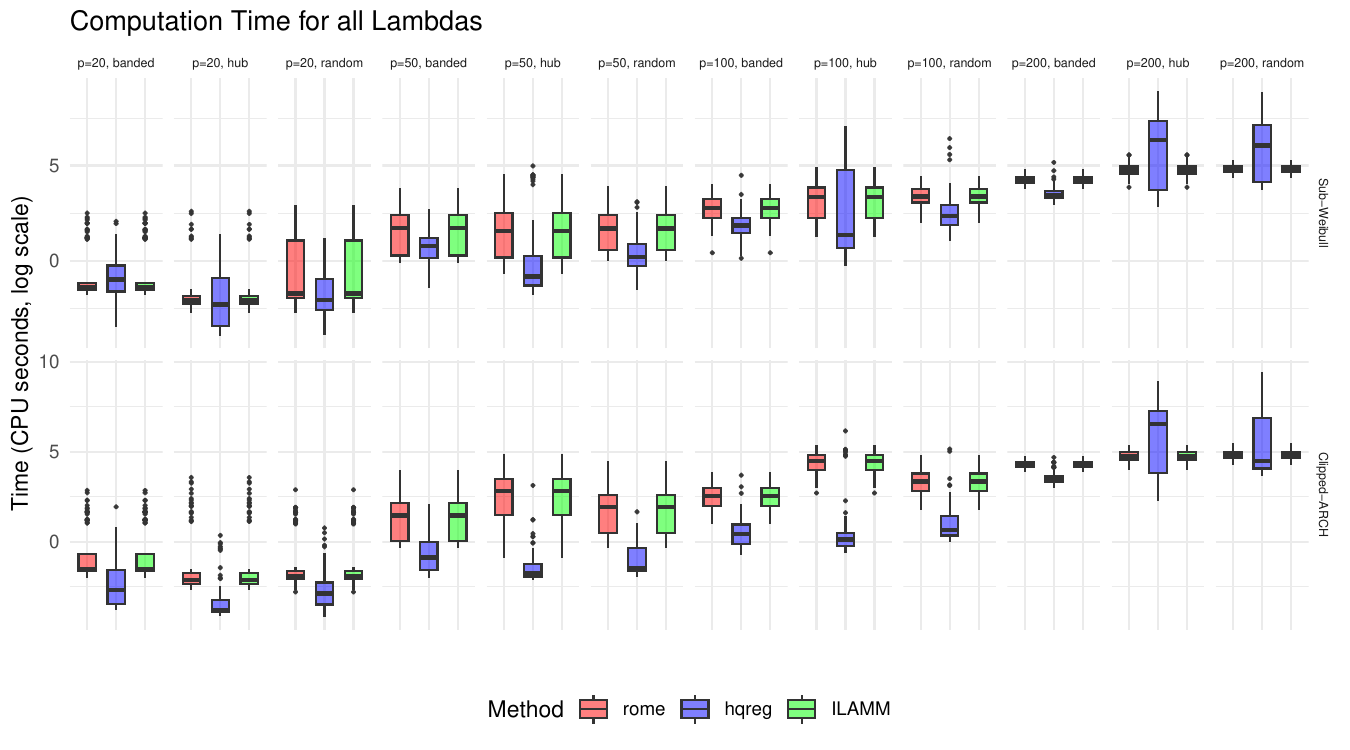}
\includegraphics[width=1\textwidth,height=0.3\textheight]
{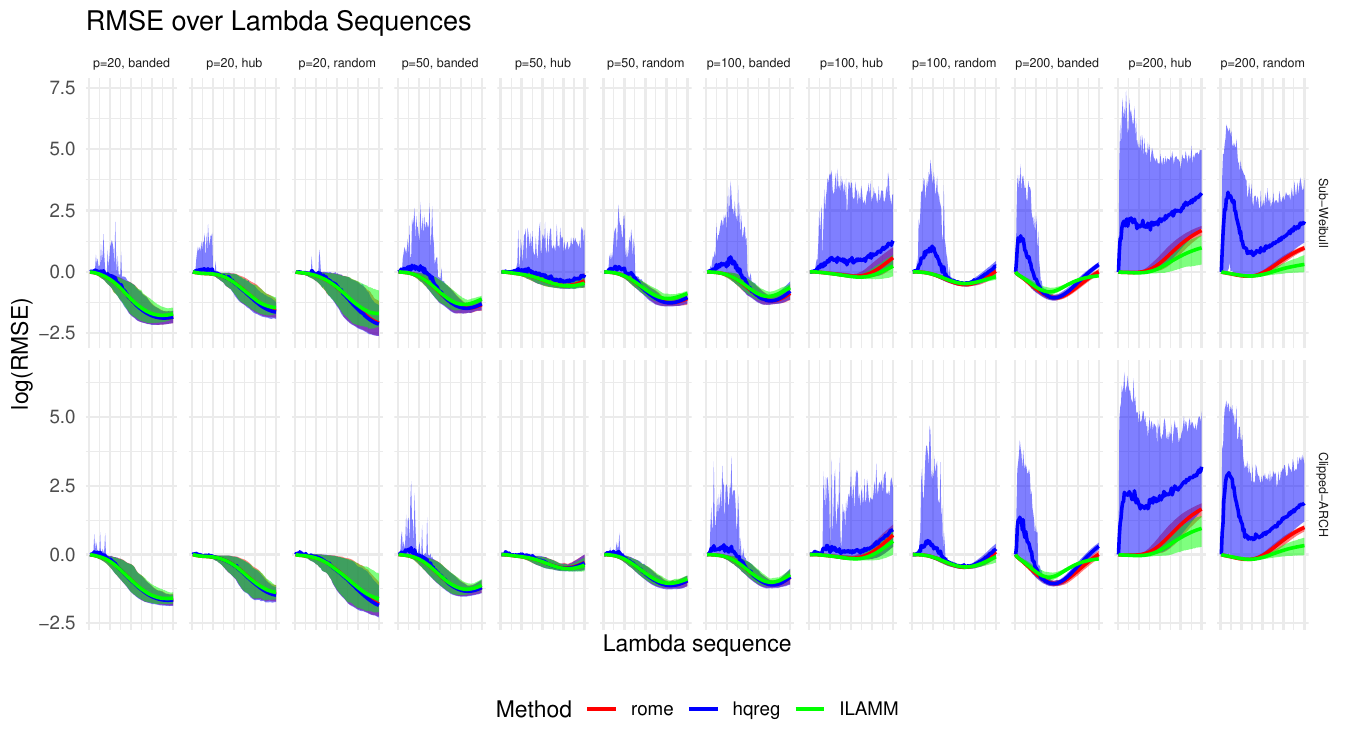}
\end{center}
\caption{Performance comparison of the proposed method (\texttt{rome}) with \texttt{hqreg} and \texttt{ILAMM}. The upper panel displays boxplots of CPU time for the full regularization path, and the lower panel shows the average normalized RMSE over 100 replications, with shaded areas representing the 5$^{\text{th}}$ and 95$^{\text{th}}$ percentile ranges. Each column corresponds to a different $(n,p)$ setting, and each row corresponds to a different VAR structure and innovation distribution.}
\label{fig:numerical_rmse_time}
\end{figure}

\clearpage
\section{Proof of Proposition \ref{prop:convergence}}
\label{ap:proofs_convergence}

In this section, we provide a complete proof of Proposition \ref{prop:convergence} in Section \ref{sse:conv}. Throughout, we use the superscript $(k,j)$, $k\geq0,j=1,\ldots,p$, notation, where $k$ denotes the iteration and $j$ the coordinate just updated. 

We follow the proof strategy of \cite{saha2013nonasymptotic_supp}. The argument proceeds in three steps.

\textbf{Step 1: Bounding the coordinate slope $\tau_{j,k}$.} The proof requires two assumptions: (i) the operator $\beta \mapsto \beta-\nabla f(\beta)/L_H$ is monotonically increasing for all coordinates in $\beta\in\mathbb{R}^p$; and (ii) the coordinate slope $\tau_{j,k}$, defined as the secant
\begin{equation}\label{e:secant}
    \tau_{j,k} 
    = \frac{g_j'(\beta_j^{(k,j)}) - g_j'(\beta_j^{(k,j-1)})}{\beta_j^{(k,j)} - \beta_j^{(k,j-1)}},
\end{equation}
is bounded above by $L_H$ for all $j=1,\ldots,p$ and $k\geq1$, where
\begin{displaymath}
    g_j'(\beta_j) 
    = \left[\nabla f\!\left(\beta_1^{(k,j-1)},\ldots,\beta_{j-1}^{(k,j-1)},\beta_j,\beta_{j+1}^{(k,j-1)},\ldots,\beta_p^{(k,j-1)}\right)\right]_j
\end{displaymath}
is the partial derivative of $f$ with respect to $\beta_j$, with all other coordinates fixed at their most recently updated values $\beta_h^{(k,j-1)}$, $h\neq j$. Assumption (i) holds under the assumption of $L_H$-Lipschitz continuity of $\nabla f$ in Step 2, which guarantees that each component of $\beta\mapsto\beta-\nabla f(\beta)/L_H$ is nondecreasing, regardless of the penalty. For assumption (ii), define
\begin{displaymath}
    g_j(z) 
    = \frac{1}{n}\sum_{i=1}^n\rho_\delta\left(y_i - \sum_{h\neq j}X_{ih}\beta_h^{(k,j-1)} - X_{ij}z\right).
\end{displaymath}
Since $\rho_\delta''(u)=\mathbf{1}_{\{|u|\leq\delta\}}\in\{0,1\}$, the second derivative satisfies
\begin{displaymath}
    g_j''(z) 
    = \frac{1}{n}\sum_{i=1}^n X_{ij}^2\mathbf{1}_{\{|r_i(z)|\leq\delta\}} 
    \leq \frac{1}{n}\sum_{i=1}^n X_{ij}^2 
    \leq \frac{1}{n}\mathrm{eig}_{\max}(X^{\top}X) = L_H.
\end{displaymath}
By the mean value theorem, there exists $\xi$ between $\beta_j^{(k,j)}$ and $\beta_j^{(k,j-1)}$ such that $\tau_{j,k}=g_j''(\xi)\leq L_H$, independently of $j$ and $k$. Note that $\tau_{j,k}$ plays the role of the secant approximation to the true curvature of $g_j$, and unlike the Lasso case, where the coordinate slope is constant, it varies with both $j$ and $k$ due to the piecewise quadratic nature of the Huber loss.

\textbf{Step 2: Lipschitz bound on $\nabla f$.} For any $\beta,\beta'\in\mathbb{R}^p$, since 
$|\rho_\delta'(u)-\rho_\delta'(v)|\leq|u-v|$ for all $u,v\in\mathbb{R}$,
\begin{align*}
    \|\nabla f(\beta)-\nabla f(\beta')\|_2 
    &= \frac{1}{n}\|X^{\top}(\psi(r(\beta))-\psi(r(\beta')))\|_2 \\
    &\leq \frac{1}{n}\|X^{\top}\|_2\|X\|_2\|\beta-\beta'\|_2 
    = \frac{1}{n}\mathrm{eig}_{\max}(X^{\top}X)\|\beta-\beta'\|_2,
\end{align*}
where $\nabla f(\beta)=-\frac{1}{n}X^{\top}\psi(r(\beta))$, and $\psi(u) =(\rho_{\delta}'(u_1),\ldots,\rho_{\delta}'(u_n))^{\top}\in\mathbb{R}^n$. Hence, $L_H\leq\frac{1}{n}\mathrm{eig}_{\max}(X^{\top}X)$.

\textbf{Step 3: Comparison of GD, CGD, and CD.} Under assumptions (i) and (ii), we show that starting from the same initial point, the objective values of the coordinate descent (CD) iterates are always no larger than those of the coordinate gradient descent (CGD) iterates, which are in turn no larger than those of the gradient descent (GD) iterates. Since the GD iterates satisfy the standard $O(1/k)$ convergence rate \cite[Theorem 3.1 in][]{beck2009fast_supp}, the CD iterates, which correspond to our proposed algorithm, satisfy \eqref{e:convergence_guarantee}.

Denote by $\{u^{(k)}\}$, $\{v^{(k,j)}\}$, and $\{\beta^{(k,j)}\}$ the solution paths of the regularized Huber regression \eqref{e:huber_regression} obtained via GD, CGD, and CD, respectively. 

\begin{itemize}
    \item \textbf{GD:} At each iteration $k$, GD updates the full vector as
    \begin{displaymath}
        u^{(k+1)} = S_{\lambda\alpha/(L_H+\lambda(1-\alpha))}\left(\frac{u^{(k)} -  \frac{1}{L_H}\nabla f(u^{(k)})}{1+\lambda(1-\alpha)/L_H}\right).
    \end{displaymath}
    This algorithm is introduced in Section \ref{se:intro}.
    \item \textbf{CGD:} At each iteration $k$, CGD updates the $j^{\textrm{th}}$ coordinates by following the one-dimensional overapproximation of $F(v)$ in \eqref{e:huber_regression} over $v_j$:
    \begin{align*}
        v_{j}^{(k,j)}&\in \argmin_{v_j\in\mathbb{R}}\Bigg\{ f(v^{(k,j-1)}) + \lambda\left(\sum_{h\neq j}|v_{h}^{(k,j-1)}| + \frac{(1-\alpha)}{2}\sum_{h\neq j}(v_{h}^{(k,j-1)})^2 \right) \\
        & + [\nabla f(v^{(k,j-1)})]_j(v_j-v_{j}^{(k,j-1)}) + \frac{L_H}{2}(v_j-v_{j}^{(k,j-1)})^2 + \lambda\left(\alpha|v_j| + \frac{(1-\alpha)}{2}v_j^2\right)\Bigg\} \\
        &= S_{\lambda\alpha/(L_H+\lambda(1-\alpha))}\left(\frac{v_j^{(k,j-1)} - \frac{1}{L_H}[\nabla f(v^{(k,j-1)})]_j}{1+\lambda(1-\alpha)/L_H}\right).
    \end{align*}
    \item \textbf{CD:} At each iteration $k$, CD updates the coordinates by exact minimization:
    \begin{align*}
        \beta_1^{(k,1)} 
        &\in \argmin_{\beta_1\in\mathbb{R}} 
        F(\beta_1, \beta_2^{(k,0)},\ldots, \beta_p^{(k,0)}), \\
        &\vdots \\
        \beta_p^{(k,p)} 
        &\in \argmin_{\beta_p\in\mathbb{R}} 
        F(\beta_1^{(k,1)},\ldots,\beta_{p-1}^{(k,p-1)},\beta_p),
    \end{align*}
    where $\beta^{(k,0)}:=\beta^{(k-1,p)}$. This is precisely the update performed by the algorithm in Section \ref{sse:regression}, solving \eqref{e:huber_regression_partial}.
\end{itemize}
For brevity, we focus on the $\ell_1$-penalty case ($\alpha=1$) in the lemmas below. The elastic net case ($\alpha\in(0,1)$) follows by replacing the soft-thresholding operator $S_{\lambda/\tau}(z_j)$ with the scaled and thresholded operator
\begin{displaymath}
    S_{\lambda\alpha/(\eta+\lambda(1-\alpha))}\left(\frac{z_j}{1+\lambda(1-\alpha)/\eta}\right),
\end{displaymath}
as derived in \eqref{e:gradient_descent} with a multiplication of the scalar $\frac{1}{1+\lambda(1-\alpha)/\eta}$ in the operator. All comparison arguments carry over verbatim under this substitution.

\begin{definition} A vector $z\in\mathbb{R}^p$ is a \emph{supersolution} (\emph{subsolution}, respectively) of \eqref{e:huber_regression} if 
and only if, for some $\tau>0$,
\begin{align}
    z 
    &\geq S_{\lambda/\tau}\left(z - \frac{\nabla f(z)}{\tau}\right), \label{e:supersolution} \\
    z 
    &\leq S_{\lambda/\tau}\left(z - \frac{\nabla f(z)}{\tau}\right). \label{e:subsolution}
\end{align}
If equality holds, i.e., $z=S_{\lambda/\tau}(z-\nabla f(z)/\tau)$, then $z$ is a minimizer of \eqref{e:huber_regression}, and conversely, since from the KKT conditions,
\begin{displaymath}
    0 \in [\nabla f(z)]_j + \lambda\partial|z_j|, \quad j=1,\ldots,p.
\end{displaymath}
In particular, any minimizer of \eqref{e:huber_regression} is automatically both a supersolution and a subsolution.
\end{definition}

\begin{lemma}\label{lem:coordinate_sufficient} Let $f$ be the smooth part of \eqref{e:huber_regression} with $L_H$-Lipschitz gradient, where $L_H\leq\frac{1}{n}\mathrm{eig}_{\max}(X^{\top}X)$. A sufficient condition for $z\in\mathbb{R}^p$ to be a supersolution is that $[\nabla f(z)]_j \geq \lambda$ for all $j$; a sufficient condition for $z$ to be a subsolution is that $[\nabla f(z)]_j \leq -\lambda$ for all $j$.
\end{lemma}
\begin{proof}
Fix $j$. If $[\nabla f(z)]_j \geq \lambda$, then
\begin{displaymath}
    z_j - \frac{[\nabla f(z)]_j}{L_H} 
    \leq z_j- \frac{\lambda}{L_H}.
\end{displaymath}
Since $S_{\lambda/L_H}(t)$ is nondecreasing and satisfies $S_{\lambda/L_H}(t) \leq t$ for all $t$, we have
\begin{displaymath}
    S_{\lambda/L_H}\left(z_j - \frac{[\nabla f(z)]_j}{L_H}\right) 
    \leq S_{\lambda/L_H}\left(z_j - \frac{\lambda}{L_H}\right) 
    \leq z_j - \frac{\lambda}{L_H} 
    \leq z_j.
\end{displaymath}
Thus, $z\geq S_{\lambda/L_H}(z - \nabla f(z)/L_H)$, so $\beta$ is a supersolution. For the subsolution case, if $[\nabla f(z)]_j \leq -\lambda$, then $z_j - [\nabla f(z)]_j/L_H \geq z_j + \lambda/L_H$, and since $S_{\lambda/L_H}(t) \geq t$ for $t \geq \lambda/L_H$, we obtain $S_{\lambda/L_H}(z_j - [\nabla f(z)]_j/L_H) \geq z_j$, so $z$ is a subsolution.
\end{proof}
From Steps 1 and 2, the Huber loss has a global Lipschitz constant $L_H$, and the secants for each $k$ and $j$ are uniformly bounded by that constant. Together with these results, Lemma \ref{lem:coordinate_sufficient} implies that the supersolution and subsolution concepts are likewise nondegenerate in the case of the Huber loss; consequently, the statement applied to GD and CGD in the Lasso case can be extended to the Huber case. Hence, we focus on the statements for CD.

The following are the technical lemmas to prove the key results. 
\begin{lemma}[cf.\ Lemmas 4.3, 4.4, 7.1, 8.1, and 8.2 in \cite{saha2013nonasymptotic_supp}]\label{lem:key_lemmas}\quad
\begin{enumerate}
    \item[(a)] If \eqref{e:supersolution} holds for some $\tau>0$, then it holds for all $\tau>0$. A similar result holds for \eqref{e:subsolution}.
    \item[(b)] If $\beta$ is a supersolution (subsolution, respectively), then for any $j=1,\ldots,p$, the map
    \begin{displaymath}
        \tau 
        \mapsto S_{\lambda/\tau}\left(\beta_j - \frac{[\nabla f(\beta)]_j}{\tau}\right)
    \end{displaymath}
    is monotonically nondecreasing (nonincreasing, respectively) on $(0,\infty)$.
    \item[(c)] Fix $k\geq0$ and $j=1,\ldots,p$. For the CD updates, 
    let
    \begin{displaymath}
        g_j(\beta_j) 
        := f\left(\beta_1^{(k,j-1)},\ldots,\beta_{j-1}^{(k,j-1)},\beta_j,\beta_{j+1}^{(k,j-1)},\ldots,\beta_p^{(k,j-1)}\right).
    \end{displaymath}
    The coordinate minimizer $\hat{\beta}_j^{(k,j)}$ from \eqref{e:huber_regression_partial} can be written as
    \begin{displaymath}
        \hat{\beta}_j^{(k,j)} 
        = S_{\lambda/\tau_{j,k}}\left(\beta_j^{(k,j-1)} - \frac{g_j'(\beta_j^{(k,j-1)})}{\tau_{j,k}}\right),
    \end{displaymath}
    where $\tau_{j,k}$ is the coordinate slope defined in \eqref{e:secant}.
    \item[(d)] If $\beta$ is a supersolution and $\beta\leq\beta'$, or $\beta$ is a subsolution and $\beta\geq\beta'$, then $F(\beta)\leq F(\beta')$.
\end{enumerate}
\end{lemma}

\begin{proof}
The proofs of (a), (b), and (d) remain the same as those in Lemmas 4.3, 4.4, 8.1, and 8.2 of \cite{saha2013nonasymptotic_supp}. For (c), note that $\hat{\beta}_j^{(k,j)}$ is the minimizer of 
$g_j(\beta_j)+\lambda|\beta_j|$, and $\tau_{j,k}$ defined in \eqref{e:secant} satisfies $\tau_{j,k} \leq L_H$. The KKT condition for this univariate problem gives
\begin{equation}\label{e:KKT_partial}
    0 \in g_j'(\hat{\beta}_j^{(k,j)}) + \lambda\,\mathrm{sign}(\hat{\beta}_j^{(k,j)}) 
    = \tau_{j,k}(\hat{\beta}_j^{(k,j)} - \beta_j^{(k,j-1)}) + g_j'(\beta_j^{(k,j-1)}) + 
    \lambda\,\mathrm{sign}(\hat{\beta}_j^{(k,j)}).
\end{equation}
If $\hat{\beta}_j^{(k,j)}=0$, then from \eqref{e:KKT_partial}, we have
\begin{align*}
    &-\tau_{j,k}\beta_j^{(k,j-1)} + g_j'(\beta_j^{(k,j-1)}) - \lambda 
    \leq 0 \leq 
    -\tau_{j,k}\beta_j^{(k,j-1)} + g_j'(\beta_j^{(k,j-1)}) + \lambda \\
    \Leftrightarrow \quad
    &\beta_j^{(k,j-1)} - \frac{g_j'(\beta_j^{(k,j-1)})}{\tau_{j,k}} - \frac{\lambda}{\tau_{j,k}} 
    \leq 0 \leq 
    \beta_j^{(k,j-1)} - \frac{g_j'(\beta_j^{(k,j-1)})}{\tau_{j,k}} + \frac{\lambda}{\tau_{j,k}},
\end{align*}
which implies $0 = S_{\lambda/\tau_{j,k}}\!\left(\beta_j^{(k,j-1)} - g_j'(\beta_j^{(k,j-1)})/\tau_{j,k}\right)$. If $\hat{\beta}_j^{(k,j)}>0$, from $g_j'(\hat{\beta}_j^{(k,j)}) + 
\lambda = 0$, we have
\begin{align*}
    \hat{\beta}_j^{(k,j)} 
    &= \beta_j^{(k,j-1)} - \frac{g_j'(\beta_j^{(k,j-1)})}{\tau_{j,k}} - \frac{\lambda}{\tau_{j,k}} 
    = S_{\lambda/\tau_{j,k}}\left(\beta_j^{(k,j-1)} - \frac{g_j'(\beta_j^{(k,j-1)})}{\tau_{j,k}}\right).
\end{align*}
The case $\hat{\beta}_j^{(k,j)}<0$ follows analogously.
\end{proof}

The following is the key result from the solutions of GD, CGD, and CD.
\begin{lemma}[cf.\ Propositions 5.1, 6.1, and Lemma 7.2 in \cite{saha2013nonasymptotic_supp}]\label{lem:GD_CGD_CD}
Assume that the operator $z \mapsto z - \nabla f(z)/L_H$ is monotonically increasing.
\begin{enumerate}
    \item[(a)] If $u^{(0)}$ is a supersolution (subsolution, respectively) and $\{u^{(k)}\}_{k\geq1}$ is generated by GD, then for all $k\geq0$,
    \begin{displaymath}
        u^{(k+1)} \leq u^{(k)} \ \& \ u^{(k)} \ \text{is a supersolution} \quad (u^{(k+1)} \geq u^{(k)} \ \& \ u^{(k)} \ \text{is a subsolution, respectively}).
    \end{displaymath}
    \item[(b)] If $v^{(0,0)}$ is a supersolution (subsolution, respectively) and $\{v^{(k,j)}\}$ is generated by CGD, then for all $k\geq0$ and $j=1,\ldots,p$,
    \begin{displaymath}
        v^{(k+1,p)} \leq v^{(k,p)} \ \& \ v^{(k,p)} \ \text{is a supersolution} \ (v^{(k+1,p)} \geq v^{(k,p)} \ \& \ v^{(k,p)} \ \text{is a subsolution, respectively}).
    \end{displaymath}
    \item[(c)] If $\beta^{(0,0)}$ is a supersolution (subsolution, respectively) and $\{\beta^{(k,j)}\}$ is generated by CD, then for all $k\geq0$ and $j=1,\ldots,p$,
    \begin{displaymath}
        \beta^{(k+1,p)} \leq \beta^{(k,p)} \ \& \ \beta^{(k)} \ \textrm{is a supersolution} 
        \ (\beta^{(k+1,p)} \geq \beta^{(k,p)} \ \& \ \beta^{(k,p)} \ \textrm{is a subsolution, respectively}).
    \end{displaymath}
    \end{enumerate}
\end{lemma}
\begin{proof}
The proofs of (a) and (b) are identical to those of Propositions 5.1 and 6.1 in \cite{saha2013nonasymptotic_supp}. For (c), we proceed by induction on $j$. Consider the case $j=1$. Since $\beta^{(k,0)}$ is a supersolution,
\begin{displaymath}
    \hat{\beta}_1^{(k,1)} 
    = S_{\lambda/\tau_{1,k}}\left(\beta_1^{(k,0)} -  \frac{g_1'(\beta_1^{(k,0)})}{\tau_{1,k}}\right) 
    \leq S_{\lambda/L_H}\left(\beta_1^{(k,0)} - \frac{g_1'(\beta_1^{(k,0)})}{L_H}\right) 
    \leq \beta_1^{(k,0)},
\end{displaymath}
where the first inequality uses statement (b) of Lemma \ref{lem:key_lemmas} and $\tau_{1,k}\leq L_H$, and the second uses the supersolution property of $\beta^{(k,0)}$. Hence, $\beta_1^{(k,1)}\leq\beta_1^{(k,0)}$ and $\beta_{j'}^{(k,1)}=\beta_{j'}^{(k,0)}$ for $j'>1$, so $\beta^{(k,1)}\leq\beta^{(k,0)}$.

For the inductive step, assume $\beta^{(k,j-1)} \leq \beta^{(k,0)}$ coordinatewise. By the same argument, we have
\begin{displaymath}
    \hat{\beta}_j^{(k,j)} 
    = S_{\lambda/\tau_{j,k}}\left(\beta_j^{(k,j-1)} - \frac{g_j'(\beta_j^{(k,j-1)})}{\tau_{j,k}}\right) 
    \leq S_{\lambda/L_H}\left(\beta_j^{(k,j-1)} - \frac{g_j'(\beta_j^{(k,j-1)})}{L_H}\right) \leq \beta_j^{(k,j-1)},
\end{displaymath}
so $\beta^{(k,j)} \leq \beta^{(k,j-1)} \leq \beta^{(k,0)}$. Next, from \eqref{e:KKT_partial} with $g_j'(\hat{\beta}_j^{(k,j)})+\lambda=0$, we have
\begin{align*}
    & \beta_j^{(k,j)} - \frac{g_j'(\beta_j^{(k,j)})}{\tau_{j,k}} 
    = \beta_j^{(k,j-1)} - \frac{g_j'(\beta_j^{(k,j-1)})}{\tau_{j,k}}, \\
    \Rightarrow
    & S_{\lambda/\tau_{j,k}}\left(\beta_j^{(k,j)} - \frac{g_j'(z_j^{(k,j)})}{\tau_{j,k}}\right) 
    = \beta_j^{(k,j)}.
\end{align*}
For $j'\neq j$, since $\beta_{j'}^{(k,j)}=\beta_{j'}^{(k,j-1)}$ and $z^{(k,j-1)} \geq \beta^{(k,j)}$ coordinatewise, we have
\begin{displaymath}
    g_{j'}'(\beta^{(k,j-1)}) - g_{j'}'(\beta^{(k,j)}) 
    \leq L_H\|\beta^{(k,j-1)} - \beta^{(k,j)}\|_2=0,
\end{displaymath}
which implies $g_{j'}'(\beta^{(k,j-1)}) \leq g_{j'}'(\beta^{(k,j)})$. Therefore,
\begin{displaymath}
    \beta_{j'}^{(k,j)} 
    = \beta_{j'}^{(k,j-1)} 
    \geq S_{\lambda/\tau_{j',k}}\left(\beta_{j'}^{(k,j-1)} - \frac{g_{j'}'(\beta^{(k,j-1)})}{\tau_{j',k}}\right) 
    \geq S_{\lambda/\tau_{j',k}}\left(\beta_{j'}^{(k,j)} - \frac{g_{j'}'(\beta^{(k,j)})}{\tau_{j',k}}\right).
\end{displaymath}
This holds for all $j'=1,\ldots,p$. Using $\tau_{j',k}\leq L_H$ and statement (b) of Lemma \ref{lem:key_lemmas}, we have
\begin{displaymath}
    \beta^{(k,j)} \geq S_{\lambda/L_H}\left(\beta^{(k,j)} - \frac{\nabla f(\beta^{(k,j)})}{L_H}\right),
\end{displaymath}
so $\beta^{(k,j)}$ is a supersolution. The subsolution case follows analogously, completing the induction.
\end{proof}

The following result allows us to compare the solutions of GD, CGD, and CD that satisfy the given conditions.
\begin{lemma}[cf.\ Theorems 6.2 and 7.3 in \cite{saha2013nonasymptotic_supp}]\label{lem:comparison}
Let $u^{(0)}, v^{(0,0)}, z^{(0,0)}$ and the sequences $\{u^{(k)}\}$, $\{v^{(k,p)}\}$, $\{\beta^{(k,p)}\}$ be as in Lemma \ref{lem:GD_CGD_CD}.
\begin{enumerate}
    \item[(a)] If $u^{(0)} = v^{(0,0)}$ are supersolutions, then for all $k\geq0$, $v^{(k,p)}\leq u^{(k)}$. Conversely, if $u^{(0)} = v^{(0,0)}$ are subsolutions, then $v^{(k,p)}\geq u^{(k)}$.
    \item[(b)] If $v^{(0,0)} = z^{(0,0)}$ are supersolutions, then for all $k\geq0$, $\beta^{(k,p)}\leq v^{(k,p)}$. Conversely, if $v^{(0,0)}=\beta^{(0,0)}$ are subsolutions, then $\beta^{(k,p)}\geq v^{(k,p)}$.
    \end{enumerate}
\end{lemma}
\begin{proof}
The proof of (a) is identical to that of Theorem 6.2 in \cite{saha2013nonasymptotic_supp}. For (b), fix $k$ and note that for each coordinate $j$, using $\tau_{j,k}\leq L_H$ and statement (b) of Lemma \ref{lem:key_lemmas}, we have
\begin{align*}
    \hat{\beta}_j^{(k,j)}
    &= S_{\lambda/\tau_{j,k}}\!\left(\beta_j^{(k,j-1)} - 
    \frac{[\nabla f(\beta^{(k,j-1)})]_j}{\tau_{j,k}}\right) 
    \leq S_{\lambda/L_H}\left(\beta_j^{(k,j-1)} - \frac{[\nabla f(\beta^{(k,j-1)})]_j}{L_H}\right) \\
    &\leq S_{\lambda/L_H}\left(v_j^{(k,j-1)} - \frac{[\nabla f(v^{(k,j-1)})]_j}{L_H}\right) 
    = v_j^{(k,j)},
\end{align*}
where the second inequality uses $\beta^{(k,j-1)} \leq v^{(k,j-1)}$, established by induction over $j$ with the base case $\beta^{(k,0)} = v^{(k,0)}$. The proof is completed by induction over $k$, following the remaining steps of Theorem 7.3 in \cite{saha2013nonasymptotic_supp}.
\end{proof}

Finally, the following result directly implies the convergence of the proposed method. Proposition \ref{prop:convergence} then follows by applying Lemma \ref{lem:main_theorem} with $z^{(0,0)} = \beta^{(0)}$ and identifying $\beta^{(k)} \equiv z^{(k,p)}$.
\begin{lemma}[cf.\ Theorem 8.3 in \cite{saha2013nonasymptotic_supp}]\label{lem:main_theorem}
Starting from the same initial point $u^{(0)}=v^{(0,0)}=\beta^{(0,0)}$, the solution paths produced by GD, CGD, and CD satisfy, for all $k\geq1$,
\begin{equation}\label{e:comparisons}
    F(\beta^{(k,p)}) 
    \leq F(v^{(k,p)}) 
    \leq F(u^{(k)}) \leq F(\beta^*) + \frac{L_H\|u^{(0)} - \beta^*\|_2^2}{2k},
\end{equation}
where $\beta^*$ is any minimizer of \eqref{e:huber_regression}.
\end{lemma}
\begin{proof}
The first and second inequalities follow by combining statement (b) of Lemma \ref{lem:comparison} with statement (d) of Lemma \ref{lem:key_lemmas}, and statement (a) of Lemma \ref{lem:comparison} with statement (d) of Lemma \ref{lem:key_lemmas}, respectively. The last inequality is the standard $O(1/k)$ convergence rate of GD \cite[see also Theorem 3.1 in][]{beck2009fast_supp}.
\end{proof}

\section{Proof of Proposition \ref{prop:KKT}}
\label{ap:proofs_kkt}

\begin{proof}
(a) For a given $\beta_j$, split the sum in the derivative of $F(\beta_j)$ according to whether $i\in\mathcal{I}(\beta_j)$ or $i\notin\mathcal{I}(\beta_j)$. Observations outside $\mathcal{I}(\beta_j)$ contribute the saturated value $\frac{\delta}{n}|X_{ij}|\sign(\beta_j-r_{ij})$ to the derivative. Summing these contributions and absorbing the sign gives the constant term $S_0 = -\frac{\delta}{n}\sum_{i=1}^n |X_{ij}|$, from which \eqref{e:slope} follows immediately.

(b) Between any two consecutive kinks $v_m<v_{m+1}$, the active set $\mathcal{I}(\beta_j)$ is constant, so $F'(\beta_j;\mathcal{I})$ is 
affine in $\beta_j$ on $(v_m,v_{m+1})$ with slope
\begin{displaymath}
    \frac{d}{d\beta_j}F'(\beta_j;\mathcal{I}) 
    = \frac{1}{n}\sum_{i\in\mathcal{I}}X_{ij}^2 + \lambda(1-\alpha) \geq 0,
\end{displaymath}
showing that $F'(\cdot\,;\mathcal{I})$ is monotonically nondecreasing. As $\beta_j$ crosses the left endpoint $r_{i_s j} - \delta/|X_{i_s j}|$ of observation $i_s$, the active set expands from $\mathcal{I}$ to $\mathcal{I}' = \mathcal{I}\cup\{i_s\}$, so $|\mathcal{I}'\setminus\mathcal{I}|=1$ and the cumulative slope increases by $X_{i_s j}^2/n \geq 0$. Since $F'(\beta_j;\mathcal{I})$ is continuous and affine on $[v_m,v_{m+1}]$, the condition $F'(v_m;\mathcal{I})F'(v_{m+1};\mathcal{I}')<0$ implies that 
$F'(\cdot;\mathcal{I})$ changes sign on $(v_m,v_{m+1})$. By the Intermediate Value Theorem, there exists $\hat{\beta}_j\in(v_m,v_{m+1})$ such that $F'(\hat{\beta}_j;\mathcal{I})=0$. The slope $A=\frac{1}{n}\sum_{i\in\mathcal{I}'}X_{ij}^2+\lambda(1-\alpha)>0$ is strictly positive since $\lambda(1-\alpha)\geq0$ and $\mathcal{I}'\neq \emptyset$, so $F'(\cdot;\mathcal{I})$ is strictly increasing on $(v_m,v_{m+1})$ and $\hat{\beta}_j$ is unique. Solving $F'(\hat{\beta}_j;\mathcal{I}) = F'(v_m;\mathcal{I})+A(\hat{\beta}_j-v_m)=0$ gives
\begin{displaymath}
    \hat{\beta}_j = v_m-\frac{F'(v_m;\mathcal{I})}{A}.
\end{displaymath}

(c) The minimizer $\hat{\beta}_j$ of \eqref{e:huber_regression_partial} satisfies the optimality condition $0\in \partial F(\hat{\beta}_j)$, where $\partial F$ denotes the subdifferential of $F$ with respect to $\beta_j$. Since the Huber loss and the ridge term are both differentiable, the subdifferential reduces to
\begin{displaymath}
    \partial F(\beta_j) = F'(\beta_j;\mathcal{I}(\beta_j)) + \lambda\alpha\partial|\beta_j|,
\end{displaymath}
where $F'(\beta_j;\mathcal{I}(\beta_j))$ is given by \eqref{e:slope} and
$\partial|\beta_j|$ is the subdifferential of the absolute value. We consider two cases: First, if $\hat{\beta}_j\neq 0$, then the optimality condition $0\in\partial F(\hat{\beta}_j)$ is equivalent to
\begin{displaymath}
    F'(\hat{\beta}_j;\mathcal{I}(\hat{\beta}_j)) + \lambda\alpha\sign(\hat{\beta}_j) = 0,
\end{displaymath}
which gives $F'(\hat{\beta}_j;\mathcal{I}(\hat{\beta}_j)) = -\lambda\alpha\sign(\hat{\beta}_j)$, so $|F'(\hat{\beta}_j;\mathcal{I}(\hat{\beta}_j))| = \lambda\alpha$. Hence \eqref{e:KKT_check} holds if and only if $\hat{\beta}_j$ is the minimizer, and no update is required when it holds. Second, if $\hat{\beta}_j=0$, then $\partial|0|=[-1,1]$, so the optimality condition $0\in\partial F(0)$ becomes
\begin{displaymath}
    0 \in F'(0;\mathcal{I}(0)) + \lambda\alpha s,\quad s\in[-1,1],
\end{displaymath}
where $\mathcal{I}(0)=\{i:|r_{ij}|\leq\delta/|X_{ij}|\}$ and the ridge term vanishes since $(1-\alpha)\hat{\beta}_j=0$. This holds if and only if $|F'(0;\mathcal{I}(0))|\leq\lambda\alpha$, which is \eqref{e:KKT_check} evaluated at $\hat{\beta}_j=0$.
\end{proof}

\end{document}